\documentclass[letterpaper,USenglish,authorcolumns,numberwithinsect]{lipics-v2019}

\newcommand{\hide}[1]{}
\hideLIPIcs
\nolinenumbers

    \newdimen\origiwspc
    \newdimen\origiwstr

\usepackage{backnaur}

\newcommand\vartextvisiblespace[1][.5em]{%
 \makebox[#1]{%
 \kern.07em
 \vrule height.3ex
 \hrulefill
 \vrule height.3ex
 \kern.07em
 }
}
\newcommand{\smallunderscore}{\texttt{\vartextvisiblespace[.7em]}}
\newcommand{\gbar}{\: | \:}
\newcommand{\gvee}{\: \vee \:}
\newcommand{\gdot}{\: \bullet \:}
\newcommand{\spanner}[1]{{\llbracket #1 \rrbracket }}

\DeclareRobustCommand{\DefVar}[1]{\emph{#1}}
\let\oldnl\nl
\newcommand{\nonl}{\renewcommand{\nl}{\let\nl\oldnl}}
 
\usepackage{mathtools} 
\usepackage[export]{adjustbox} 

\let\origdoublepage\cleardoublepage
\newcommand{\clearemptydoublepage}{%
 \clearpage{\pagestyle{empty}\origdoublepage}}
\let\cleardoublepage\clearemptydoublepage
\usepackage{array}
\newcolumntype{L}{>{\centering\arraybackslash}m{6cm}}

\usepackage{tikz}
\usetikzlibrary{automata, positioning, arrows, shapes}

\usepackage{color,soul} 

\DeclareRobustCommand{\BK}[1]{{\sethlcolor{magenta}\hl{B: #1}}}

\usepackage{multirow} 
\usepackage{lipsum} 
\usepackage{caption}
\usepackage{amsfonts}

\usepackage{subcaption} 
\usepackage{stmaryrd} 

\let\oldnl\nl

\usepackage [lined,boxed,ruled]{algorithm2e}

\SetCommentSty{xCommentSty}

\DeclareMathOperator*{\join}{\bowtie}
\newcommand{\allen}[1]{{\Gamma}_{\!(#1)}}

\usepackage{thmtools}
\usepackage[T1]{fontenc}
\usepackage{hyperref}
\usepackage[framemethod=tikz]{mdframed}
\usepackage{etoolbox}

\newlength\mylen

\AtBeginDocument{%
 }

	\definecolor{dartmouthgreen}{rgb}{0.05, 0.5, 0.06}
\definecolor{carrotorange}{rgb}{0.93, 0.57, 0.13}
\definecolor{dodgerblue}{rgb}{0.12, 0.56, 1.0}
\newcounter{examplec}
\counterwithin*{examplec}{section}
\renewenvironment{example}[1][]{
\stepcounter{examplec}
\par\vspace{5pt}\noindent
\fbox{\textbf{Example~\thesection.\theexamplec}}
\hrulefill\par\vspace{10pt}\noindent\rmfamily}

\newcommand{\autour}[1]{\tikz[baseline=(X.base)]\node
[draw=teal,semithick,rectangle,inner sep=2pt, rounded corners=3pt,
scale=0.8] (X) {\strut #1};}
\newcommand{\regtext}[1]{{\texttt{\autour{{#1}}}}}
\newcommand{\regtextinl}[1]{{\small\texttt{\autour{{#1}}}}}

\usepackage{listings}

\lstdefinestyle{JInLnStyle}{
 basicstyle=\small\ttfamily,
 frame=none,
 tabsize=4, aboveskip=10pt, belowskip=10pt, %
lineskip=2pt
}
\lstnewenvironment{InLnJRule}{\lstset{style=JInLnStyle}}{}

\usepackage{color}
\definecolor{gray}{rgb}{0.4,0.4,0.4}
\definecolor{darkblue}{rgb}{0.0,0.0,0.6}
\definecolor{cyan}{rgb}{0.0,0.6,0.6}

\lstdefinelanguage{XML}
{
 morestring=[b]",
 morestring=[s]{>}{<},
 morecomment=[s]{<?}{?>},
 stringstyle=\color{black},
 identifierstyle=\color{darkblue},
 keywordstyle=\color{cyan},
 morekeywords={xmlns,version,type}
}

\AtBeginDocument{%
 }

\title{Improving Unstructured Data Quality via Updatable Extracted Views}

\author{Besat Kassaie}
{David R. Cheriton School of Computer Science, University of Waterloo, Waterloo, Ontario, Canada, N2L 3G1}{bkassaie@uwaterloo.ca}{}{}

\author{Frank Wm. Tompa}
{David R. Cheriton School of Computer Science, University of Waterloo, Waterloo, Ontario, Canada, N2L 3G1}{fwtompa@uwaterloo.ca}{}{}

\authorrunning{B. Kassaie and F. W. Tompa}
\Copyright{Besat Kassaie and Frank Wm. Tompa}


\begin{document}

\maketitle

\begin{abstract}
 Improving data quality in unstructured documents is a long-standing challenge. Unstructured data, especially in textual form, inherently lacks defined semantics, which poses significant challenges for effective processing and for ensuring data quality. We propose leveraging information extraction algorithms to design, apply, and explain data cleaning processes for documents. Specifically, for a simple document update model, we identify and verify a set of sufficient conditions for rule-based extraction programs to qualify for inclusion in our document cleaning framework. Through experiments conducted on medical records, we demonstrate that our approach provides an effective framework for identifying and correcting data quality problems, thereby highlighting its practical value in real-world applications.
\end{abstract}

\section{Introduction}
Data cleaning represents a significant expense within data processing pipelines~\cite{costofdatacleaning}. Recent research efforts have focused extensively on devising effective methods for cleaning structured data. In contrast, collections of unstructured data, such as medical reports, tend to exhibit more inconsistencies and errors compared to structured databases. This is primarily due to the high variability of unstructured data, which makes it more prone to human error. Nevertheless, current data cleaning techniques for unstructured data are often inadequate and primarily embedded within other processing steps~\cite{5917038,DBLP:conf/dmdw/GalhardasFSSS01, DBLP:conf/cbms/DeshpandeRTFRA20, woo2019application, DBLP:conf/fdse/NguyenPNVTS20, DBLP:journals/jdiq/RoyMG18}. For every downstream task, a clean version of the data is prepared in a suitable format, such as a bag of tokens or extracted relations. The cleaned data can then be used in applications of interest, but this approach leaves the data sources uncleaned. Therefore, a similar cleaning process must be initiated from scratch for each new application. For example, Deshpande et al.~\cite{DBLP:conf/cbms/DeshpandeRTFRA20} improve the performance of an information retrieval system by cleaning unstructured medical data sources. Their cleaning process is seamlessly integrated into their application-specific pipeline and thus not easily adoptable by others. Even though they extract elements from the data sources and clean them for retrieval purposes, they do not produce cleaned documents.

In this work, we propose instead to clean documents and use the cleaned version in applications. To this end, we recommend to use document-at-a-time information extractors to create a relational view over documents. With an adequate extractor, data quality failures in documents can be revealed in extracted views, cleaning can then be applied over the extracted items, and finally the cleaned items can be transferred back into the source documents (or a copy if the original is to be preserved). To this end, we consider a document ``clean'' if its corresponding extracted view is clean. 

Although this process seems intuitive, extractors employed for this purpose must satisfy specific conditions. For example, consider correcting an error in one cell of the extracted view. After transferring the correct value back to the document, we expect to see the correct version of that specific cell when re-running the same extractor, and we expect all other values to remain unchanged (Fig.~\ref{fig:exvu}). Without such a guarantee, we might not converge to a corresponding clean state, because each update to the view might create unpredictable side-effects in the document. Therefore, the ability to predict the content of the extracted view after transferring changes to the source documents becomes imperative to maintain the predictability of the entire process. 
\begin{figure}[ht]
 \centering
\includegraphics[width=0.50\columnwidth]{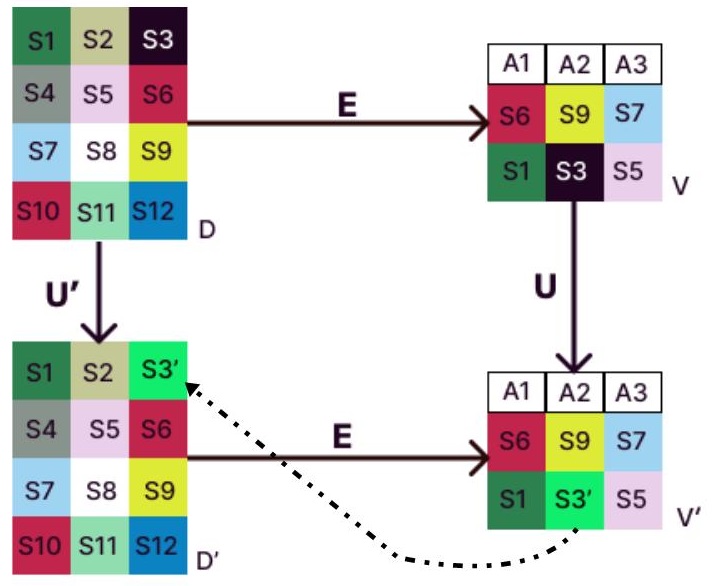}
\caption{Applying extractor $E$ to document $D$ produces table $V$. Updating $V$ to form $V'$, mapping the update back to form $D'$, and then re-applying $E$ should produce that same updated table $V'$.}
\label{fig:exvu}
\end{figure}

In this paper, we develop solutions for rule-based information extraction systems, such as GATE~\cite{cunningham2014developing} and SystemT~\cite{DBLP:conf/icde/ReissRKZV08}, that are based on the theory of regular languages.\footnote{We explain the rationale behind our decision to focus on rule-based extractors in Section~\ref{subsubsec:whyrulebase}.}
 In such systems, rules use regular expressions to describe not only values in the document corpus that are to be extracted but also strings that serve as contexts for the values to be extracted. AQL~\cite{chiticariu2010systemt}, the rule-based language underlying SystemT, benefits from a formal model, namely \textit{document spanners}~\cite{DBLP:journals/jacm/FaginKRV15}, which facilitates analysis of extraction programs.

\subsection{Novelty and Contributions}
This work is novel in two key aspects. First, it addresses the unexplored and challenging problem of document cleaning, a topic that has received limited attention in the literature despite its significance. Second, the approach proposed in this paper is innovative and has not been attempted before, offering a fresh perspective on solving similar problems.
The fundamental realization that information extraction can be viewed as a mechanism for document database management facilitates  a variety of improvements in text processing similar to the ones identified in relational database views. 
In the context of unstructured data management, this work is innovative in identifying, formulating, and addressing the challenges associated with updating extracted views. We leverage the potential of updatable extracted views to devise a solution aimed at enhancing unstructured data quality. We anticipate that future research in this under-explored field will focus on uncovering and solving critical challenges, driving significant advancements in unstructured data quality. 
To this end, 
\begin{itemize}
  \item We identify and formalize the extracted view update problem. 
  \item We formalize a general view update model, i.e., \textit{domain preserving updates} along with an intuitive update translation mechanism.
  \item We introduce a sufficient property called \textit{stability} of extraction programs, for which we are able to translate view updates to document updates.
  \item We present a verification process that can be applied to rule-based extractors to determine whether the program is stable.
  \item We show how stable extractors can be used in a document cleaning pipeline, presenting experimental results that demonstrate how stable extractors can be used to improve the quality of medical documents.
\end{itemize}

\section{Extracted Relations as Materialized Views}\label{sec:ERasMV}
We propose that an extracted relation should be treated as a materialized view of a document corpus~\cite{DBLP:conf/doceng/KassaieT20, kassaie2023update}. To this end, we require extractors to have four general characteristics: i) \DefVar{DAAT}: tuples are extracted from individual documents (``document-at-a-time''); ii) \DefVar{strict}: for every possible input document the set of extracted values in the corresponding record is a subset of words and phrases appearing in that document. That is, each extracted value comes from a continuous span in some document. Hence, a hypothetical extractor that mines the input text and infers information that does not appear explicitly in the corpus is not a strict extractor;\footnote{We leave to future research the exploration of whether strictness can be extended to inferences realized through a bijective function, mapping words or phrases in the input documents (such as birth dates) to other values present in the extracted relation (such as ages).} iii) \DefVar{computable}: for all possible input documents and corresponding extracted values, the extractor also provides the positions from which values are extracted. Some extraction mechanisms directly extract positions from text while considering corresponding string values as by-products of extraction process, whereas others do the reverse; iv) \DefVar{deterministic}: for every possible set of input documents, the set of extracted values remains consistent across multiple runs of the extractor; some extractors, such as those modelled by document spanners, are inherently deterministic, but others, such as some programs written in JAPE, may depend on implementation, time-of-day, randomization seeds, or other external characteristics. 

 \begin{figure}[ht]
 \centering
\includegraphics[width=0.7\columnwidth]{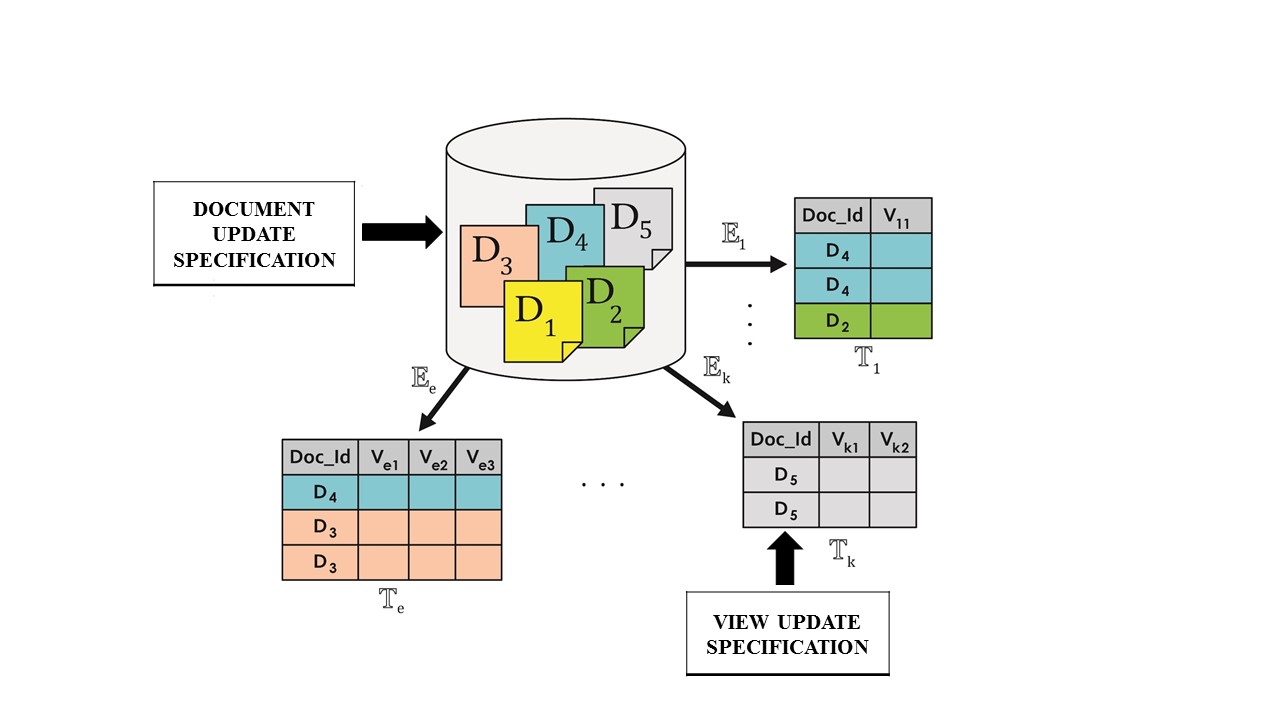}
\caption{Extraction system that supports updates to source documents as well as extracted views.} 
\label{fig:docDB}
\end{figure}

Our hypothetical extraction system comprises five components: a collection of documents  $\mathbb{D}$, a set of extraction programs $\{\mathbb{E}_1, \cdots, \mathbb{E}_e\}$ that run independently over $\mathbb{D}$, a corresponding set of extracted relations, an instance of a document update specification, and an instance of a view update specification (Figure~\ref{fig:docDB}).
The union of  relations produced by an extractor $\mathbb{E}_k$ against the document database is stored in a relation $\mathbb{T}_k$ that includes an additional column to associate each document identifier with the spans for the corresponding span relation.
 These tables serve as materialized views of the document database.

Because our extractors are deterministic and operate on one document at a time, each one defines a function from documents to relations. Because we also require extractors to be strict and computable, each cell in an extracted relation can be interpreted as a \DefVar{document span}, i.e., the start and end offsets within the document identified by the corresponding $Doc\_id$, together with the string value appearing within those offsets in the document. In practice, both the spans and their corresponding string values need not be displayed explicitly, but rather the values might be derived on demand from the spans in a \DefVar{span relation} or the spans might be hidden and only the values displayed in an \DefVar{extracted string relation}.

Treating extraction algorithms as a view mechanism allows extractors to be adopted in a broad range of applications. For instance, when extraction time is a bottleneck and updates to source documents occur frequently, it has been proposed to update extracted relations incrementally~\cite{DBLP:conf/icde/ChenDYR08} or to avoid re-extraction altogether~\cite{DBLP:journals/pvldb/KassaieT23}. 
Similarly, finding and repairing violations of constraints in a large corpus of uncleaned documents is a difficult problem: an extractor that guarantees the preservation of consistency between the source text and the extracted relations can be adopted to solve this problem by mapping a cleaned view back to the underlying documents. This latter problem is well-motivated and addressed in the relational setting~\cite{DBLP:journals/jacm/CosmadakisP84,DBLP:books/sp/kimrb85/FurtadoC85,DBLP:conf/pods/Keller85}. However, due to the diverse range of extraction techniques and inherent complexity in text processing and understanding, tackling these problems for collections of documents introduces new challenges.

\begin{figure}[ht]
    \centering
    
    \begin{subfigure}[t]{\textwidth} 
        \centering
        \begin{adjustbox}{max width=0.98\linewidth} 
            \renewcommand{\arraystretch}{1} 
            \setlength\tabcolsep{1pt} 
       \begin{tabular}{ccccccccccccccccccccccccccccccccccccccccc}

O&n&{\smallunderscore}&{{0}}& {{{3}}}&{{/}}&{{2}}&{{4}}&,&{\smallunderscore}& {w}&e&{\smallunderscore}&r&e&n&t&e&d&{\smallunderscore}&a&n&d&{\smallunderscore}&w&a&t&c&h&e&d&{\smallunderscore}&`&\textcolor{red}{\textbf{M}}&\textcolor{red}{\textbf{I}}&\textcolor{red}{\textbf{B}}&\textcolor{black}{\textbf{'}}&.&{\smallunderscore}&I&{\smallunderscore}\\
\tiny1&\tiny2&\tiny3&\tiny4&\tiny5&\tiny6&\tiny7&\tiny8&\tiny9&\tiny10&\tiny11&\tiny12&\tiny13&\tiny14&\tiny15&\tiny16&\tiny17&\tiny18&\tiny19&\tiny20&\tiny21&\tiny22&\tiny23&\tiny24&\tiny25&\tiny26&\tiny27&\tiny28&\tiny29&\tiny30&\tiny31&\tiny32&\tiny33&\tiny34&\tiny35&\tiny36&\tiny37&\tiny38&\tiny39&\tiny40&\tiny41\\
h&i&g&h&l&y&{\smallunderscore}&r&e&c&o&m&m&e&n&d&{\smallunderscore}&i&t&{\smallunderscore}&a&s&{\smallunderscore}&a&n&{\smallunderscore}&i&n&s&p&i&r&i&n&g&{\smallunderscore}&a&n&d&{\smallunderscore}&h\\
\tiny42&\tiny43&\tiny44&\tiny45&\tiny46&\tiny47&\tiny48&\tiny49&\tiny50&\tiny51&\tiny52&\tiny53&\tiny54&\tiny55&\tiny56&\tiny57&\tiny58&\tiny59&\tiny60&\tiny61&\tiny62&\tiny63&\tiny64&\tiny65&\tiny66&\tiny67&\tiny68&\tiny69&\tiny70&\tiny71&\tiny72&\tiny73&\tiny74&\tiny75&\tiny76&\tiny 77&\tiny78&\tiny79&\tiny80&\tiny81&\tiny82\\
u&m&o&r&o&u&s&{\smallunderscore}&f&i&l&m&{\smallunderscore}&t&o&{\smallunderscore}&e&n&j&o&y&.\\
\tiny83&\tiny84&\tiny85&\tiny86&\tiny87&\tiny88&\tiny89&\tiny90&\tiny91&\tiny92&\tiny93&\tiny94&\tiny95&\tiny96&\tiny97&\tiny98&\tiny99&\tiny100&\tiny101&\tiny102&\tiny103&\tiny104\\

 \end{tabular}
        \end{adjustbox}
       \caption{Original document.}
      \label{fig:toyexample1A}
    \end{subfigure}
    \hfill 
~
    \begin{subfigure}[t]{\textwidth} 
        \centering
        \begin{adjustbox}{max width=0.98\linewidth} 
            \renewcommand{\arraystretch}{1}
            \setlength\tabcolsep{1pt}
      \begin{tabular}{ccccccccccccccccccccccccccccccccccccccccc}
O&n&{\smallunderscore}&0&3&/&2&4&,&{\smallunderscore}&w&e&{\smallunderscore}&r&e&n&t&e&d&{\smallunderscore}&a&n&d&{\smallunderscore}&w&a&t&c&h&e&d&{\smallunderscore}&`&\textcolor{dartmouthgreen}{\textbf{M}}&\textcolor{dartmouthgreen}{e}&\textcolor{dartmouthgreen}{\textbf{n}}&\textcolor{dartmouthgreen}{{\smallunderscore}}&\textcolor{dartmouthgreen}{\textbf{i}}&\textcolor{dartmouthgreen}{\textbf{n}}&\textcolor{dartmouthgreen}{\smallunderscore}&\textcolor{dartmouthgreen}{\textbf{B}}\\
\tiny1&\tiny2&\tiny3&\tiny4&\tiny5&\tiny6&\tiny7&\tiny8&\tiny9&\tiny10&\tiny11&\tiny12&\tiny13&\tiny14&\tiny15&\tiny16&\tiny17&\tiny18&\tiny19&\tiny20&\tiny21&\tiny22&\tiny23&\tiny24&\tiny25&\tiny26&\tiny27&\tiny28&\tiny29&\tiny30&\tiny31&\tiny32&\tiny33&\tiny34&\tiny35&\tiny36&\tiny37&\tiny38&\tiny39&\tiny40&\tiny41\\
\textcolor{dartmouthgreen}{\textbf{l}}&\textcolor{dartmouthgreen}{\textbf{a}}&\textcolor{dartmouthgreen}{\textbf{c}}&\textcolor{dartmouthgreen}{\textbf{k}}&'&{\textbf{.}}&{{\smallunderscore}}&I&\smallunderscore&{h}&i&g&h&l&y&\smallunderscore&r&e&c&o&m&m&e&n&d&\smallunderscore&i&t&{\smallunderscore}&a&s&{\smallunderscore}&a&n&{\smallunderscore}&i&n&s&p&i&r\\
\tiny42&\tiny43&\tiny44&\tiny45&\tiny46&\tiny47&\tiny48&\tiny49&\tiny50&\tiny51&\tiny52&\tiny53&\tiny54&\tiny55&\tiny56&\tiny57&\tiny58&\tiny59&\tiny60&\tiny61&\tiny62&\tiny63&\tiny64&\tiny65&\tiny66&\tiny67&\tiny68&\tiny69&\tiny70&\tiny71&\tiny72&\tiny73&\tiny74&\tiny75&\tiny76&\tiny 77&\tiny78&\tiny79&\tiny80&\tiny81&\tiny82\\
i&n&g&\smallunderscore&a&n&d&\smallunderscore&h&u&m&o&r&o&u&s&\smallunderscore&f&i&l&m&\smallunderscore&t&o&{\smallunderscore}&e&n&j&o&y&.\\
\tiny83&\tiny84&\tiny85&\tiny86&\tiny87&\tiny88&\tiny89&\tiny90&\tiny91&\tiny92&\tiny93&\tiny94&\tiny95&\tiny96&\tiny97&\tiny98&\tiny99&\tiny100&\tiny101&\tiny102&\tiny103&\tiny104&\tiny105&\tiny106&\tiny107&\tiny108&\tiny109&\tiny110&\tiny111&\tiny112&\tiny113\\
 \end{tabular}

        \end{adjustbox}
        \caption{Updated document.}
     \label{fig:toyexample1C}
    \end{subfigure}

    \caption{A sample input document and its updated version, with associated offsets indicated beneath each character, starting from $1$. The substrings in \textbf{\textcolor{red}{red}} undergo updates and those highlighted in \textbf{\textcolor{dartmouthgreen}{green}} represent new values.}
    \label{fig:toyexample1}
\end{figure}
\section{Motivating Example}
\label{exmp:motivating1}
We wish to replace any shortened movie name in a document collection with the corresponding full title before releasing the data. For simplicity, we consider  any string  a movie name if it is surrounded by  \regtextinl{`} \regtextinl{'} \footnote{Throughout this paper, characters to be matched in a document are represented \regtext{like this} to distinguish them from other text.} and appears in a sentence with terms such as \regtextinl{watched}  or \regtextinl{saw}, where words are considered to be in the same sentence if there are no end punctuation characters such as \regtextinl{?} or \regtextinl{.} between them. 

Let $\Sigma$ represent the set of Latin alphanumeric characters, punctuation and space characters (the last represented by \regtextinl{\smallunderscore}). The extraction semantics can be represented using a regular expression with variables. We use variables, such as $A$ (for Action) and $M$ (for Movie) in this example, to mark each matched substring:
\small
$$\gamma_{mv}=\Sigma^* \regtext{\smallunderscore}A\{(\regtext{watched} \vee \regtext{saw})\}\gamma_{b}\gamma_{s}^*\regtext{\smallunderscore`}M\{(\gamma_u \vee \gamma_l)\gamma_{m}^*\}\regtext{'} (\gamma_{b} \vee \regtext{?} \vee \regtext{.})\Sigma^*$$

\normalsize
\noindent where $\gamma_{u} = \footnotesize [\regtextinl{A},\regtextinl{Z}]$, $\gamma_{l}= \footnotesize [\regtextinl{a},\regtextinl{z}]$, $\gamma_b = \regtextinl{\smallunderscore} \vee \regtextinl{,}$, $\gamma_{m}=\gamma_u \vee \gamma_l \vee \regtextinl{\smallunderscore}$, and $\gamma_{s} = \Sigma - \regtextinl{?} - \regtextinl{.}$. In this paper, we adhere to the ``all matches'' semantics for regular expression matching, and for this example we assume that string values are assigned to variables and get extracted for every match.

After applying the extractor to the input text in Figure~\ref{fig:toyexample1A}, we replace the movie name \regtextinl{MIB}\footnote{https://en.wikipedia.org/wiki/Men\_in\_Black\_(film\_series)} in the extracted string relation (Figure~\ref{fig:toyexampleext1CString}) with its complete form \regtextinl{Men\smallunderscore in\smallunderscore Black} resulting in Figure~\ref{fig:toyexampleext1DString}. 

Now, if both the extracted relation and the source documents are to be two representations of the same information, they need to stay consistent. Consequently, when updates are made to an extracted relation, a translation mechanism is required to reflect new values back into the source document. In this paper, we adopt an intuitive approach to the translation process: we replace the old values with the new values in their corresponding positions within the source document.

After applying this simple update to the source document in Figure~\ref{fig:toyexample1A}, we get the document in Figure~\ref{fig:toyexample1C}. Re-running the extractor over the updated document would produce the updated relation in Figure~\ref{fig:toyexampleext1DString}, that is: I) \regtextinl{Men\smallunderscore in\smallunderscore Black} is extracted, II) no new rows appear, III) no rows disappear, and IV) no other cells' values have changed. 
It is important to note that in practical scenarios, we often work with complex extractors that involve multiple, and potentially conflicting, expressions. In this work, our objective is to establish whether, for a given extractor, it is possible to apply any valid update to the relation (whether to improve data quality, effect anonymization, or for any other reason) and reflect the update back in the source document, regardless of the remaining content of the source document.

\begin{figure*}[t!]
 \centering
 
\begin{subfigure}[t]{0.45\linewidth}
 \centering
\begin{tabular}{cc}
 A&M\\
 \hline
watched &MIB\\
\end{tabular} 
 \caption{Original extracted string relation.}
 \label{fig:toyexampleext1CString}
\end{subfigure}
~
\begin{subfigure}[t]{0.45\linewidth}
 \centering
\begin{tabular}{cc}
 A&M\\
 \hline
watched & Men \smallunderscore in \smallunderscore Black\\
\end{tabular} 
 \caption{Updated string relation.}
 \label{fig:toyexampleext1DString}
\end{subfigure}

 \caption{Extracted relation and its updated version for motivating example.}
 \label{fig:toyexampleext1}
\end{figure*}

In our document cleaning framework, users design cleaning-specific stable extractors (see Section~\ref{sec:stableExtr}) to identify items of interest corresponding to data quality concerns\footnote{Data cleaning is a specific task within the broader field of data quality. In this paper, we use the terms "data cleaning" and "data quality" interchangeably unless otherwise specified.} within the document database. Based on these extractors, the extracted values, and possibly some predefined data cleaning rules, the user applies a cleaning process to the extracted view. Finally, the system translates the updates made in the extracted view back to the original documents (Figure~\ref{fig:cleaning}). This framework introduces an update translation mechanism (Step 3) that translates updates applied to extracted relations (Step 2) into updates to the original documents. A verification process based on static program analysis conducted beforehand ensures the feasibility of this translation. Users responsible for data quality may use any cleaning algorithms on the relations (Step 2) for extracted relations as long as they respect the conditions set forth by the framework (see Section~\ref{sec:updateModel} and~\ref{def:respect} in Section~\ref{sec:stableExtr}).
 \begin{figure}[ht]
 \centering
\includegraphics[width=0.7\columnwidth]{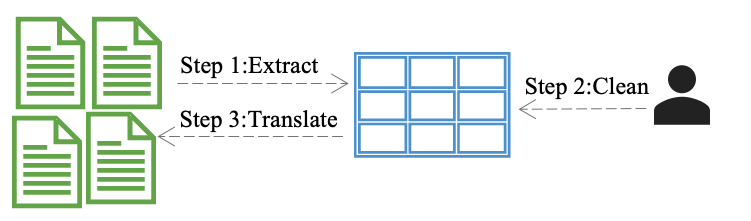}
\caption{ Overview of Proposed Document Cleaning Framework. } 
\label{fig:cleaning}
\end{figure}

\section{Computational model for extractors}

Fagin et al.~\cite{DBLP:journals/jacm/FaginKRV15} define a formal model for representing extractors written in AQL~\cite{chiticariu2010systemt}. In this section, we summarize the basics that we need to describe our work here.\footnote{The text in this section is substantially the same as what can be found in our previous formulations of extractors~\cite{DBLP:journals/pvldb/KassaieT23,kassaie2023update, DBLP:conf/doceng/KassaieT20}.}

\subsection{Regular Expressions with Capture Variables}\label{sec:regex}
Given a finite alphabet $\Sigma$, a regular expression extended using variables chosen from a set $V$ is called a \DefVar{regex with capture variables} and conforms to $\gamma$ in the grammar $G_S(\Sigma,V)$ as follows:
\begin{equation}
\begin{aligned}
\gamma\: & :=\: \varnothing \gbar\epsilon \gbar \alpha \gbar (\gamma \gvee \gamma) \gbar (\gamma \gdot \gamma) \gbar (\gamma)^{*} \gbar x\{\gamma\} \\
\alpha\: & :=\: \sigma \gbar [\sigma,\sigma] \gbar \delta \\
\delta\: & :=\: \Sigma \gbar (\delta - \sigma)
\end{aligned}
\label{eq:reg_exp_var}
\end{equation}
\noindent where $\sigma$ represents any character in $\Sigma$, $[\sigma,\sigma]$ represents the disjunction of characters having their encoding between or equal to the encodings of the first and second character in the range, the terminal symbol $\Sigma$ represents the disjunction of all characters in $\Sigma$, $\delta - \sigma$ represents the disjunction of all characters in $\delta$ with the exception of $\sigma$, and (what distinguishes these expressions from conventional regular expressions) $x$ represents any variable in $V$.
Given a regex with capture variables $r$, the corresponding \DefVar{regex tree} $\mathcal{T}(r)$ represents the hierarchical structure of $r$, in which the tree's leaves have labels $\varnothing$, $\epsilon$, characters in $\Sigma$, character ranges in $\Sigma$, or the character $\Sigma$ itself, and internal nodes have labels $\bullet$, $\vee$, $*$, $-$, or a symbol in $V$. 
For convenience of notation, when writing a regex with capture variables, we follow common practice for regular expressions in allowing the following shorthand: omission of parentheses (relying instead on left associativity of all operations and precedence of $-$ over $*$ over $\bullet$ over $\vee$) and omission of the operator $\bullet$. 

If $E$ is a regex with capture variables, then we denote the set of capture variables in $E$ as $\mathit{SVars}(E)$. The use of a subexpression of the form $n\{g\}$ in $E$ signifies that whenever $E$ matches a string, the span containing a substring matched by $g$ is to be \emph{marked by the capture variable $n$}. We distinguish two classes of variables  based on their relative positioning: a variable  $x \in \mathit{SVars}(E)$ is \DefVar{exposed} if it is not enclosed in any other variables, otherwise it is \DefVar{nested}. In this paper, we assume that every regex with a capture variable has at least one exposed variable.

\subsection{Document Spans}
A document $D$ is a finite string over some alphabet: $D \in \Sigma^{*}$. A \DefVar{span} of document $D$, denoted $[i, j \rangle$ ($1 \leq i \leq j \leq |D|+1$), specifies the start and end offsets of a substring in $D$, which is in turn denoted $D_{[i, j \rangle}$, and extends from offset $i$ through offset $j-1$.  $[i, i \rangle$ denotes an empty span at offset $i$. Spans $s_1 = [i_1, j_1 \rangle$  and $s_2 = [i_2, j_2 \rangle$ are identical if and only if $i_1=i_2$ and $j_1=j_2$.

\begin{table*}
    \caption{Allen's interval relationships extended to spans.}
    \label{tab:Allen}
   \begin{minipage}{\columnwidth}
    \begin{center}\small
       \begin{tabular}{rlc@{}l}
        1&$\allen{X < Y}$ & \emph{X precedes Y} & \multirow{2}{*}{$ \bigg \} \Sigma^{*}X \{ \Sigma^{*}\} \Sigma^{+} Y \{ \Sigma^{*} \} \Sigma^{*}$}\\
    
        2&$\allen{Y > X}$ & \emph{Y is preceded by X} &\multirow{2}{*}{} \\
    &&&\\
    3&$\allen{X \mathbf{m} Y}$ & \emph{X meets Y} & \multirow{2}{*}{$\bigg\}\Sigma^{*}X\{\Sigma^{+}\}Y\{\Sigma^{+}\}\Sigma^{*}$}\\
        4&$\allen{Y \mathbf{mi} X}$ & \emph{Y is met by X} &\multirow{2}{*}{}  \\
          &&&\\
        5&$\allen{X \mathbf{o} Y}$ & \emph{X overhangs Y} & \multirow{2}{*}{$\bigg\}\Sigma^{*}(X \vdash)\Sigma^{+}(Y \vdash)\Sigma^{+}(\dashv X)\Sigma^{+}(\dashv Y)\Sigma^{*}$}\footnote{This abuse of notation represents a spanner (as described in Section~\ref{sec:Extractors}) matching an automaton with operators that open and close the variables $X$ and $Y$ as indicated.}\\
        6&$\allen{Y \mathbf{oi} X}$ & \emph{Y is overhung by X} &\multirow{2}{*}{} \\
          &&&\\
        7&$\allen{X \mathbf{d} Y}$ & \emph{X during Y} & \multirow{2}{*}{$\bigg\}\Sigma^{*}Y\{\Sigma^{+}X\{\Sigma^{*}\}\Sigma^{+}\}\Sigma^{*}$}\\
        8&$\allen{Y \mathbf{di} X}$ & \emph{Y contains X}  \\
          &&&\\
        9&$\allen{X \mathbf{s} Y}$ & \emph{X starts Y} & \multirow{2}{*}{$\bigg\}\Sigma^{*}Y\{X\{\Sigma^{*}\}\Sigma^{+}\}\Sigma^{*}$}\\
        10&$\allen{Y \mathbf{si} X}$ & \emph{Y is started by X}  \\
          &&&\\
        11&$\allen{X \mathbf{f} Y}$ & \emph{X finishes Y} & \multirow{2}{*}{$\bigg\}\Sigma^{*}Y\{\Sigma^{+}X\{\Sigma^{*}\} \}\Sigma^{*}$}\\
        12& $\allen{Y \mathbf{fi} X}$ & \emph{Y is finished by X}  \\
          &&&\\
          13&$\allen{X = Y}$ & \emph{X is equal to Y} & {$\Sigma^{*}X\{Y\{\Sigma^{*}\}\}\Sigma^{*} \gvee \Sigma^{*}X\{\epsilon\}Y\{\epsilon\}\Sigma^{*}$}\\
   \end{tabular}
  
\end{center}
\end{minipage}
\end{table*}

\subsection{Verification} \label{sec:verU2V}
A substantial portion of our work is based on investigating various relationships between spans. 
Allen has defined a set of 13 possible relationships between non-empty intervals~\cite{DBLP:journals/cacm/Allen83}. These can be extended to capture the same basic relationships among spans (including empty spans) as summarized in Table~\ref{tab:Allen}. All possible relationships among spans can be described by disjunctions of these basic relationships; for example, ``X overlaps Y'' ($\allen{X\cap Y}$) can be expressed as the disjunction of the last nine basic relationships.\footnote{The definition of overlapping spans given by Fagin et al.~\cite{DBLP:journals/jacm/FaginKRV15} is asymmetric for empty spans, i.e., given a span $[i,j\rangle$ the empty span at $[i,i\rangle$ is considered overlapping with $[i,j\rangle$ while the empty span at  $[j,j\rangle$ is considered disjoint from $[i,j\rangle$. We treat both as overlapping.}  ``X overlaps but is not equal to Y'' ($\allen{X\Cap Y}$) can be similarly expressed as the disjunction of the fifth through the twelfth basic relationships.

\subsection{Rule-based Extractors Modeled as Document Spanners} \label{sec:Extractors}

Determining membership in a language defined by a regex with capture variables $E$ can be accomplished by executing a corresponding \DefVar{vset-automaton} $\mathcal{A}(E)$. 

Given $E$ and $V = \mathit{SVars}(E)$, $\mathcal{A}(E)$ is a non-deterministic finite state automaton augmented with a designated set (initially empty) and two operators for each variable $x \in V$, namely $x \! \vdash$ (``open $x$'') and $\dashv \! x$ (``close $x$''). 
Besides including standard character transitions, $\mathcal{A}(E)$ also includes \emph{operation transitions} that, instead of consuming a character from the input string, insert the variable $x$ into the designated set if the transition label is $x \! \vdash$ and remove $x$ from the designated set if the label is $\dashv \! x$. A document $D$ is accepted by $\mathcal{A}(E)$ if, after scanning the whole input, we end up in a final state and 
the designated set is empty. A \DefVar{matching} of $E$ against document $D$ is an accepting run in $\mathcal{A}(E)$, where for each variable $x \in \mathit{SVars}(E)$, the spans marked by $x$ each begin with the offset in $D$ when $x$ is inserted into the designated set and end with the offset in $D$ when $x$ is removed from that set. 

If $E$ is a regex with capture variables, it specifies a \DefVar{document spanner}, denoted as $\spanner{E}$, with $\mathit{SVars}(\spanner{E}) = \mathit{SVars}(E)$. Applying a document spanner to a document $D$ produces a \DefVar{span relation}, i.e., a relation that contains spans of $D$. Thus, $\spanner{E}$ is a function mapping strings over $\Sigma^{*}$ to  $\mathcal{S}^{|\mathit{SVars}(E)|}$ where $\mathcal{S}$ is the set of all spans of $D$. To ensure that the span relation is in first-normal form with no null values, we restrict our attention to a specific class of document spanners, namely \DefVar{functional document spanners}, that mark exactly one span for each variable for all accepting runs, regardless of the input document $D$. In particular, for a given document $D$, the spanner specified by $E$ produces a span relation $\spanner{E}(D)$ in which there is one column for each variable from $\mathit{SVars}(E)$ appearing in $E$, each row corresponds to a matching of $E$ against $D$ when the variables are ignored, and the value in a row for the column corresponding to $x \in \mathit{SVars}(E)$ is the span marked by $x$.

\begin{definition}Throughout this paper, a functional document spanner used for the purpose of information extraction is called an \DefVar{extraction spanner}. An extraction spanner is specified by an algebraic expression over \DefVar{extraction formulas}, each of which uses a \DefVar{regular expression with capture variables} to mark which parts of a matched document are to be extracted into a span relation, which is called an \DefVar{extracted relation}. 
\end{definition} \label{def:extraction_spanner}

\begin{example} 
Let $\Sigma$ be the set of Latin alphanumeric, punctuation, and space characters (the last represented by \regtextinl{\smallunderscore}). Note that this is the setting for all future examples unless otherwise stated. $\gamma_{date}$ is a regular expression with capture variables\: \small
\[\gamma_{date} =\Sigma^*\regtext{mdate="}\:F\{Y\{\gamma_{d}\gamma_{d}\gamma_{d}\gamma_{d}\}\regtext{-}M\{\gamma_{d}\gamma_{d}\}\regtext{-}D\{\gamma_{d}\gamma_{d}\}\}\regtext{"}\Sigma^*\] \normalsize
where $\gamma_{d}=[\regtextinl{0},\regtextinl{9}]$ and $\mathit{SVars}(E)=\{F,Y,M,D\}$. $F$ is the only exposed variable in $\mathit{SVars}(E)$. For every successful matching of $\gamma_{date}$ against a string, a sub-string matching $\gamma_{d}\gamma_{d}\gamma_{d}\gamma_{d}$ is marked by the capture variable $Y$. We extend conventional set notation to write $Y \subset F$ and $\regtextinl{-} \in F$.

\end{example}\label{exmp:gamma-date}

Let $S$, $S_1$, and  $S_2$ be extraction spanners where the last two are union-compatible (i.e., $\mathit{SVars}(S_1) = \mathit{SVars}(S_2)$); $X \subseteq \mathit{SVars}(S)$; and $x,y \in \mathit{SVars}(S)$. An algebra over spanners can be defined with operators: 
\begin{enumerate}
    \item union: $S_1 \cup S_2$ having variables $\mathit{SVars}(S_1 \cup S_2) = \mathit{SVars}(S_1)$, 
    \item projection: $\pi_X(S)$ having variables $X$, 
    \item natural join: $S \join S_1$ having variables $\mathit{SVars}(S) \cup \mathit{SVars}(S_1)$, and 
    \item binary string selection: $\zeta^{=}_{x,y} S$  having variables $\mathit{SVars}(S)$. 
\end{enumerate}
The set of \DefVar{core spanners} (corresponding to the core of SystemT's AQL) includes extraction spanners specified by any extraction formula  together with all extractionspanners in the closure of core spanners under this algebra. Applying a core spanner to any document $D$ is equivalent to applying each included primitive spanner to $D$ and then applying the corresponding relational operators to the extracted relations.  The set of core spanners defines the extraction spanners subject to analysis in this paper.

\subsection{Efficient Construction of Extraction Spanners}\label{subsec:EfficientConstruction}
Given one or more spanners as input, we use various algebraic operations defined over spanners to verify properties of those spanners statically.
Specifically, we convert the inputs to \DefVar{eVset-automata}, a variant of vset-automata with the same expressivity, as proposed by Morciano~\cite{morciano2017engineering}. By preserving three properties of eVset-automata (namely, \DefVar{well-behaved}, \DefVar{pruned}, and \DefVar{operation-closed}),  Morciano is able to construct eVset-automata in polynomial time for simulating the application of projection, union, and join over spanners.
His thesis also shows that converting a regex with capture variables to an automaton and checking for emptiness can be done in polynomial time. It is trivial to show that renaming variables can also be accomplished in polynomial time with eVset-automata.

\section{Characterization of Stable Information Extraction Programs}\label{sec:properties}

In this section, we define stable information extraction programs. We show that if an information extraction program is stable, we can alter the document in such a way that the updated version of an extracted relation can be extracted directly from the altered text using the same extraction program. We propose a verifier that tests sufficient conditions on core AQL programs to determine whether the program is stable.

\begin{definition}
For the remainder of this paper, an \DefVar{extractor} is a program defined by an extraction spanner. Together with a function similar to AQL's built-in function \textit{GetText()}, an extractor can be applied to a document to produce a string relation specified by the spanner, while maintaining access to the corresponding spans in that document. We refer to the extracted string relation as the \DefVar{output view} of the document defined by the extractor. 
\end{definition}\label{def:extractor}

 In the core of AQL, extraction formulas can be combined using the algebraic operations projection, join, union, and string selection to form a tree\footnote{In fact, because intermediate spanners can be re-used, the operator ``tree'' for a spanner might be more generally a singly-rooted directed acyclic graph. Nevertheless, for pedagogic simplicity, we will refer to it as if it were a tree.} in which each subtree is an extraction spanner, leaves correspond to extraction formulas and internal nodes are algebraic operators. Thus, an extractor comprising $k$ extraction spanners is denoted $\mathcal{P}=\{\mathcal{P}_1, \cdots, \mathcal{P}_k\}$ and forms a tree, the root of which (without loss of generality, $\mathcal{P}_k$) is designated to be the \DefVar{output spanner}.  
Any potential updates are applied to an attribute $A_i$ in the output view, which corresponds to the capture variable $V_i$ in the output spanner. 

In general, if an extractor is defined by $\mathcal{P}=\{\mathcal{P}_1, \cdots, \mathcal{P}_k\}$, each capture variable in $\mathcal{P}_k$ might be defined in extraction formulas in multiple leaves of the corresponding tree and propogated to the root via applications of union and join. To identify all such formulas, for each variable $V_i \in \textit{SVars}(\mathcal{P}_k)$, we introduce the \textit{variable provenance set} denoted as $\mathit{Prov}(V_i)$:
$$\mathit{Prov}(V_i)=\{E| E \text{ is an extraction formula} \; \land V_i \in \mathit{SVars}(E) \land V_i \; \text{is not eliminated by projection} \}$$ 
For each variable $V_i$, $\mathit{Prov}(V_i)$ can be inferred by traversing the tree of extraction spanners, and every occurrence of any capture variable $V_i$ in any formula in $\mathit{Prov}(V_i)$ is called an \DefVar{updatable variable}. 

\subsection{Extracted View Update Model}\label{sec:updateModel}

In a string relation, $R$, we denote the domain of the $i^{th}$ attribute $A_i$ by $W_i$, that is, $R \subseteq W_1 \times ... \times W_{\mathcal{T}}$. 
Given a cell $r_i$ for an attribute $A_i$ in the extracted record $r$, our update model enables the replacement of $r_i$'s value with any valid value from $W_i$. Therefore, we define the update function as follows: Let $\mathcal{F}$ be an indexed set of \textit{domain-preserving} functions so that $\mathcal{F}=\{f_i | f_i: W_i \to W_i\}$ where $i \in [1\dots\mathcal{T}]$ and $W_i$ is the domain for attribute $A_i$.
If $r= \langle v_1,\dots ,v_\mathcal{T} \rangle$ is a record in $R$, we denote the corresponding record in which the $i^{th}$ attribute is modified as $F(r,i)= \langle v_1',\dots ,v_\mathcal{T}' \rangle$, where 
\begin{equation*}
v_k'=\begin{cases}

f_k(v_k) & \text{if $k=i$},\\
v_k &\text{otherwise} .\\

\end{cases}
\end{equation*}
Using our update model denoted as $U$, we can update a whole record by repeatedly updating its cells. Therefore, we extend this notation to let $F(r)= \langle f_1(v_1),\dots ,f_\mathcal{T}(v_\mathcal{T}) \rangle$. As described in Figure~\ref{fig:cleaning}, we propose to clean documents by cleaning their corresponding extracted relations and then transferring the resulting updates on those relations back to the underlying documents.

\subsection{Update Translation Mechanism} 
In order to ensure consistency between an extracted view and the source document, it is essential to have a mechanism for translating updates made to the view back to the source document, as has been done for relational databases. In this work, we propose a straightforward and intuitive translation mechanism: we replace the document value located at the span corresponding to the modified cell with the new value, Step 3 in Figure~\ref{fig:cleaning}.

Let $\mathcal{X}$ be a strict and computable extractor, $D$ be a document, $A_i$ be an extracted attribute with value $v_i$, and $[a_i,b_i \rangle$ denote the span in $D$ starting at the character at offset $a_i$ and ending before the character at offset $b_i$ from which $v_i$ of record $r$ is extracted.  We define a document synthesizer function $DSyn(D,r,i,f_i(v_i))$ that for an input document generates a document that results from substituting $f_i(v_i)$ for the occurrence of $v_i$ located at $[a_i,b_i \rangle$. 
Formally, 
$DSyn(D,r,i,f_i(v_i))= D_{[1, a_i \rangle} \bullet f_i(v_i) \bullet D_{[b_i, |D|+1 \rangle}$, the slice of the document up to (but not including) offset $a_i$ concatenated with the replacement value and the slice of the document starting at offset $b_i$.
\subsection{Stable Extractors} \label{sec:stableExtr}
\begin{definition}\label{def:stable}
 An information extraction algorithm $\mathcal{X}$ is \DefVar{stable} if it is strict, computable, and deterministic and for any document $D$, $\forall i \in [ 1 \dots \mathcal{T} ]$, and $\forall r \in R$ we have
$$\mathcal{X}(D)=R \implies \mathcal{X}(DSyn(D,r,i,f_i(v_i)))= \{r'|r' \in R \land r' \neq r \} \cup F(r,i)$$ where $R$ is the extracted string relation. 
\end{definition}
 
\noindent Thus, with a stable extractor, changing a value in the appropriate position in a document affects only the expected cell in the extracted relation, as illustrated in Figure~\ref{fig:exvu}. Our extraction process operates on a document-at-a-time basis, guaranteeing that records extracted from other documents cannot be affected when $D$ is altered to reflect the update to an individual cell. Therefore, we need only consider possible side-effects on those records extracted from the same source document as the one corresponding to an updated cell. 

\begin{theorem}
Consider a stable extractor $\mathcal{X}$, any indexed set of domain preserving functions $\mathcal{F}=\{f_i | f_i : W_i \to W_i$, where $i \in [ 1 \dots \mathcal{T} ] \}$, and any document $D$. For all $i \in [ 1 \dots \mathcal{T} ] $ and $r \in R$, substituting $f_i(v_i)$ for $v_i$ in $[a_i,b_i \rangle$ produces $D_\mathcal{F}$ in such a way that $F(\mathcal{X}(D))=\mathcal{X}(D_\mathcal{F})$.
\end{theorem}

\begin{proof} 
$\mathcal{X}$ is strict, therefore all $v_i$ in $r$ occur in $D$.
Being a computable extractor, the document span 
$[a_i,b_i \rangle$ is known for every cell $v_i$ in $r$. So the locations of all spans in $D$ that need to be modified is accessible to the procedure. $\mathcal{X}$ is deterministic, thus the set of extracted values remains consistent across multiple runs of the extractor.
Finally, $\mathcal{X}$ is stable, so any substitutions corresponding to each attribute $v_i$ affect only the $i^{th}$ attribute in $r$.
\end{proof}

\noindent Thus, if an extractor is stable, then it produces a string relation that allows us to translate updates made to the relation back to the source documents. We next describe a verifier that statically analyzes an extraction program to determine whether the specified extractor is stable.

When an extracted relation is updated, the value is propagated back to the document, and the relation is re-extracted, an unstable, deterministic extractor might produce a relation in which the updated value does not appear as expected, some other value in the relation is altered, or some other rows appear or disappear.
In the remainder of this section, we thoroughly examine potential causes for these possibilities and devise conditions that we require to be met by AQL progrms to guarantee stability. Additionally, we develop a diagnostic tool that, given a document spanner and a collection of domain preserving functions as our update model $U$, determines whether the spanner satisfies certain sufficient conditions that guarantee the corresponding extractor to be stable.

\subsubsection{Domain Consistency}
The view update function replaces the value in a particular cell with another value chosen from the domain of its associated attribute. Generally, the set of possible values for an extracted attribute, $A_i$, is implicitly defined by the extraction program. We require that, for each domain-preserving function in $\mathcal{F}$,
the domain of $f_i$ is a subset of the domain formed by the capture variable $V_i$. Specifically, in an AQL program an extraction formula is bound to $V_i$, which determines the possible values for $V_i$. However, we might have multiple instances of such an extraction formula in $\mathcal{P}$, since 1) each extraction formula might have multiple disjuncts each having a distinct extraction formula for $V_i$ 2) multiple extraction formulas in $\mathcal{P}$ might include $V_i$, i.e., $|\mathit{Prov}(V_i)| > 1$. 
We require that the regular (sub)expressions that are captured by any instance of $V_i$ must all be identical. With this restriction, the updated value is guaranteed to be matched by \emph{whichever} instance of $V_i$ matched the value prior to the update.

 \begin{definition}\label{def:domaincons}
A core AQL program is deemed \DefVar{domain-consistent} if $\forall v \in \mathit{SVars}(\mathcal{P}_k)$ the corresponding domains of $v$ are identical in all formulas in $\textit{Prov}(v)$. 
 \end{definition}
 Algorithm~\ref{alg:DomainConsistency} relies on the disjunctive form of the input expression.  We have shown elsewhere~\cite{DBLP:journals/pvldb/KassaieT23} that a functional extraction formula 
$E$ with an exposed variable 
$v$ can be written in a special disjunctive form, denoted as 
$\Delta(E,v)$. To construct $\Delta(E,v)$, all disjunctions in $E$ that have the  variable $v$ in their disjuncts are ``pulled up'' over concatenations in the corresponding extended regex tree to create separate disjuncts at the outermost level of the formula.
\begin{lemma}
    Algorithm~\ref{alg:DomainConsistency} determines whether $\mathcal{P}$ is domain consistent.
\end{lemma} 
\begin{proof}
The algorithm extracts all regular expressions marked by each update variable and compares them pairwise to ensure that they are equivalent.
\end{proof}

\begin{algorithm} 
\caption{Domain Consistency Test.}
\label{alg:DomainConsistency} 
\small 
\LinesNumbered
\SetKwInput{Pre}{Precondition}
 \KwIn{Core AQL program $\mathcal{P}$}
\KwOut{Boolean}
\tcc{for all updatable variables}
\ForAll{$V_i \in \mathit{SVars}(\mathcal{P}_k)$}{

$rgxLs \gets emptyList()$ \tcc{create an empty list of regular  expression}

\ForAll{$E \in \mathit{Prov}(V_i)$}{
 
 $S' \gets \pi_{V_i} E$ \tcc{project $E$ on $V_i$}
 
\tcc{for every disjunct $p$ in $S'$}
\ForAll{$p\in \Delta(S', V_i)$}{
\tcc{retrieve expression enclosed by $V_i$}
$rgxLs.add(getEnclosedRegEx(p, V_i))$ 
}
}
\ForAll{$r\in rgxLs$}{
\ForAll{$r' \in rgxLs$}{
 \If{$r \setminus r' \neq \emptyset$}{
  
  \Return{$False$}\tcc{regular expressions are not equivalent}
}
}}
}

\Return{$True$} \tcc{$\mathcal{P}$ is domain consistent}

\end{algorithm}

\normalsize
\subsubsection{Non-Conflicting Extractor} 
We need to find whether re-running an extractor over an updated document extracts \textit{correctly modified} records. 
Unwanted side-effects can be an unexpected value appearing in the updated cell, modified values for any other cells in the extracted relation, and gaining or losing one or more records in the relation. 
Such side-effects will occur when an updated span overlaps one or more other strings that are extracted or contextual text that are used to determine other strings to extract. 
In this section, we categorize the problematic overlaps and propose a mechanism to verify  sufficient conditions for a given core AQL program to be conflict-free (Definition~\ref{def:confree}).

For a core AQL program, two levels of conflicts are defined: 1) Intra-formula conflicts: These conflicts arise independently for each extraction formula within the program; 2) Inter-formula conflicts: These conflicts occur between pairs of extraction formulas. 

\noindent \textbf{Intra-formula conflicts} \\
Intra-formula conflicts occur when an extraction formula $\spanner{E}$ has problematic overlaps with its own updatable variables.\footnote{An extraction formula without any updatable variables is free of intra-formula conflicts.} We construct four spanners to examine whether an extraction spanner represented by an extraction formula might have intra-formula conflicts. 
\paragraph{Case I}
Sometimes two distinct spans extracted as the same variable include identical subspans. This requires that within the extracted relation corresponding to $\spanner{E}$, there exist at least two distinct cells, denoted $c$ and $c'$, that are extracted as $V_i$, where $E \in \text{{Prov}}(V_i)$ and $V_i \in \mathit{SVars}(\mathcal{P}_k)$. Consider the following simple example:
\begin{example}
We define $\gamma_{abbr}$ to extract mentions of drug abbreviations from a medical text:
\small $$\gamma_{abbr}=\llbracket\Sigma^* \gamma_b A\{\gamma_u\gamma_u^*(\regtext{\smallunderscore}\gamma_u\gamma_u^*)^*(\gamma_d)^*\} \gamma_b \Sigma^* \rrbracket$$ \normalsize
\noindent where again $\gamma_{u} =  [\regtextinl{A},\regtextinl{Z}]$, $\gamma_{d}=[\regtextinl{0},\regtextinl{9}]$  and $\gamma_b = \regtextinl{\smallunderscore} \vee \regtextinl{,}$. Running $\gamma_{abbr}$ on $\regtextinl{\smallunderscore ARA\smallunderscore C\smallunderscore}$ three distinct cells with overlapping subspans will be extracted as $A$ including $[1,4\rangle$, $[1, 6\rangle$, and $[5, 6\rangle$.  According  to ISMP~\cite{ISMPListofErrorProneAbbreviations},    \regtextinl{ARA\smallunderscore C} could have been mistakenly written as \regtextinl{ARA\smallunderscore A} in some instances that needs to be corrected. However, when correcting the extracted value on $[1, 6\rangle$ and transferring the corrected term back to the source document, values corresponding to more than one extracted cell, $[1, 6\rangle$ (expected) and  $[5, 6\rangle$ (not expected), would be changed.
\end{example}

Given an extraction formula $E$ and a capture variable $v \in \mathit{SVars}(E)$, we construct $CaseI(E,v)$: 
$$CaseI(E,v)= 
 \pi_X (\rho_{v \rightarrow X} \spanner{E}) \join \spanner{\allen{X\: \Cap \: Y}} \join \pi_Y ( \rho_{v\rightarrow Y} \spanner{E})$$ 
\noindent where 
 $\allen{X\Cap Y}$ is the disjunction of the fifth through the twelfth basic relationships in Table~\ref{tab:Allen}, which corresponds to spans that are unequal but share a subspan.
\begin{lemma} \label{lm:case1}
Given an extraction formula $E$ and a capture variable $v \in \mathit{SVars}(E)$, if $CaseI(E,v)=\emptyset$ then spans extracted as $v$ are non-overlapping unless they are equal. 
\end{lemma}
\begin{proof}

$\allen{X\: \Cap \: Y}$ represents the universal spanner that  contains unequal spans marked as $X$ and $Y$ where  any span that is marked by $X$  has at least one subspan in common with spans marked by $Y$. Therefore, $CaseI(E,v)$ represents all documents that match $E$ while $\pi_{v}(\spanner{E}(D))$  include at least two distinct rows contaning spans marked as $v$ that cover an identical subspan.  

\end{proof}

\paragraph{Case II}

An extracted span, might appear in more that one record as variable $v$. Consider the following example: 
\begin{example}\label{exp:case2}
Let $\gamma_{med}$ extract mentions of medications and dose designations:
\small $$\gamma_{med}=\llbracket \Sigma^* \regtext{\smallunderscore} M\{\gamma_u\gamma_u^*(\regtext{\smallunderscore}\gamma_u\gamma_u^*)^*(\gamma_d)^*\} \gamma_{b}\gamma_{s}^*D\{\gamma_{d}\gamma_{d}^*\}\regtext{\smallunderscore}(\regtext{ml} \vee \regtext{mg} )\regtext{\smallunderscore}\Sigma^* \rrbracket$$ \normalsize
\noindent where $\gamma_{s} = (\Sigma - \regtextinl{?}-\regtextinl{.})$ and the other expressions are as defined earlier. 
Applying $\gamma_{med}$ to 
\small
\begin{center}
 \regtext{Take\smallunderscore CPZ\smallunderscore140\smallunderscore mg\smallunderscore for\smallunderscore three\smallunderscore days,\smallunderscore then\smallunderscore 40\smallunderscore mg\smallunderscore for\smallunderscore two\smallunderscore weeks.}   
\end{center}
\normalsize
will extract the span for \regtextinl{CPZ} as $M$ in two records, one for each dosage. Updating the value in one of the records will also cause the other record to be changed.\footnote{Even though this might not be viewed as erroneous in this instance, it violates the definition of ``non-conflicting.''}
\end{example}

We construct the following spanner to detect such situations: 
$$CaseII(E,v, z)=
 \pi_{X, v} (\rho_{z\rightarrow X} \spanner{E}) \join \spanner{\allen{X\: \neq \: Y}} \join \pi_{Y, v} ( \rho_{z \rightarrow Y} \spanner{E}) $$
 
\noindent where $\allen{X\neq Y}$ is the disjunction of the first 12 basic relationships in Table~\ref{tab:Allen} and $z \in \mathit{SVars}(E) \setminus \{v\}$. 

\begin{lemma}\label{lm:case2}
 Given an extraction formula $E$ and $v \in \mathit{SVars}(E)$ if $\forall z \in \mathit{SVars}(E) \setminus \{v\}$ $CaseII(E,v, z)=\emptyset$, then there cannot exist two distinct rows of the extracted relation in which the spans extracted as $v$ are identical. 
\end{lemma}

\begin{proof}
Similar to the proof of Lemma~\ref{lm:case1}.
\end{proof}

\hide{
            \begin{definition}
            If $\forall v \in SVar(E)$ where $v \in V_u(\mathcal{P}_k)$ and $E \in \mathit{Prov}(v)$ we have $\forall z \in \mathit{SVars}(E) \setminus \{v\}$ $CaseII(E,v,z)=\emptyset$ then $E$ is free of \DefVar{Case II conflicts}. 
            \end{definition}
            \BK{these are not the best terms to use: $E$ is free of conflicts of Case II }
}
\paragraph{Case III}

\noindent An extracted span corresponding to a variable might have problematic overlap with spans associated with other capture variables:

\begin{example} \label{exmpl:case3}
Let $\gamma_{med}$ extract mentions of medications, dose designations, strength, and frequency:
\small $$\gamma_{med}=\llbracket \Sigma^* \regtext{\smallunderscore} M\{\gamma_u\gamma_u^*(\regtext{\smallunderscore}\gamma_u\gamma_u^*)^*(\gamma_d)^*\} \gamma_{b}^*D\{\gamma_{d}\gamma_{d}^*\}\regtext{\smallunderscore}(\regtext{ml} \vee \regtext{mg} )\gamma_{b}^*S\{\gamma_{str}?\}\gamma_{b}F\{\gamma_{\mathit{freq}}\} \Sigma^* \rrbracket$$ \normalsize
where   $\gamma_{str}=\gamma_u\gamma_u \vee \regtextinl{diluted} \vee \regtextinl{half-strength}\vee ...$\normalsize, and \small$\gamma_{\mathit{freq}}=\gamma_u\gamma_u (\regtextinl{\smallunderscore} (\gamma_u\gamma_u \vee \regtextinl{with\smallunderscore food} \vee \regtextinl{bedtime}\vee ...)^*$\normalsize.

Applying $\gamma_{med}$ to \small
\begin{center}
\regtext{....\smallunderscore ASA\smallunderscore20\smallunderscore mg\smallunderscore HS\smallunderscore OD\smallunderscore...}
\end{center} \normalsize
\noindent will mark \regtextinl{HS} as $S$ (Strength) in one row and as part of $F$ (frequency) in a different row. ISMP recommands using \regtextinl{half-strength} or \regtextinl{bedtime} instead of \regtextinl{HS}, however any modification involving \regtextinl{HS} will either affect more than one cell or cause another row to disappear. 
\end{example}

We construct the following spanner to find out whether $v$ overlaps with other capture variables: 
 $$ Case III(E,v,z) = \pi_X(  \rho_{v \rightarrow X} \spanner{E}) \join \spanner{\allen{X\: \cap \: Y}} \join ( \rho_{z \rightarrow Y} \spanner{E}) $$
where $\allen{X\: \cap \: Y}$ is the disjunction of the fifth through thirteenth basic relationships in Table~\ref{tab:Allen} and $z \in \mathit{SVars}(E) \setminus \{v\}$. 

\begin{lemma}\label{lm:case3}
 Given an extraction formula $E$ and $v \in SVar(E)$ if $\forall z \in \mathit{SVars}(E) \setminus \{v\}$ $CaseIII(E,v, z)=\emptyset$ then there exists no row in the extracted relation in which the span extracted by $v$ ovelaps with spans extracted by other capture variables. 
\end{lemma}
\begin{proof}
Similar to the proof of Lemma~\ref{lm:case1}.
\end{proof}

For an extraction formula to match a document, certain strings must appear either inside or outside the regions marked by capture variables. If an extracted string is replaced by a new string in the source document, we wish to know whether any of the strings that specify required contextual information is disrupted in any way. Elsewhere we have presented an algorithm that, given an input extraction formula $E$ with exposed variable $v$, constructs  a modified  version of $E$ denoted as $\mathcal{C}_v(E)$ in which all contextual expressions are also marked with capture variables~\cite{DBLP:journals/pvldb/KassaieT23}.
For example, the contextualized version of the extraction formula introduced in Example~\ref{exp:case2}, $\mathcal{C}_{M}(\gamma_{med})$, is constructed as: \small
$$\llbracket \Sigma^* \textcolor{red}{c_1\{}\regtext{\smallunderscore}\textcolor{red}{\}}M\{\gamma_u\gamma_u^*(\regtext{\smallunderscore}\gamma_u\gamma_u^*)^*(\gamma_d)^*\} \gamma_{b}^*D\{\gamma_{d}\gamma_{d}^*\}\textcolor{red}{c_2\{}\regtext{\smallunderscore}(\regtext{ml} \vee \regtext{mg} )\textcolor{red}{\}}\gamma_{b}^*S\{\gamma_{str}?\}\textcolor{red}{c_3\{}\gamma_{b}\textcolor{red}{\}}F\{\gamma_{freq}\}\Sigma^* \rrbracket$$ \normalsize
\noindent where  $c_1$, $c_2$, and $c_3$ are new variables created to mark contexts. 

Contextualization covers all subexpressions except for those of the form $X_i^{*}$,  where $L(X_i)$ is a regular language that has strings of length $1$ only (i.e, \DefVar{unigrams}); when contextualizing an extraction formula, we do not cover subexpressions representing arbitrary strings in the alphabet $\hat{\Sigma} \subseteq \Sigma$.   
For example, in $\gamma_{med}$, $\gamma_b^*$ and $\Sigma^*$ are of this form and therefore remain uncovered in $\mathcal{C}_{M}(\gamma_{med})$.
 
\paragraph{Case IV}
We construct the following spanner to investigate whether the updates supported by the extraction formula $E$ may modify regions of a document that are designated as context by $E$:
$$CaseIV(E,v,z)=
 \pi_{X} ( \rho_{v \rightarrow X}\ \spanner{E}) \join \spanner{\allen{X\: \cap \: Y}} \join \rho_{z \rightarrow Y} ( \mathcal{C}_{v}({E})).$$ 
 where $\allen{X\: \cap \: Y}$ is the disjunction of the fifth through thirteenth basic relationships in Table~\ref{tab:Allen} and $z\in \mathit{SVars}(\mathcal{C}_{v}(E)) \setminus \{v\}$.
\begin{lemma}\label{lm:case4}
 Given an extraction formula $E$ and $v \in \mathit{SVars}(E)$ if $\forall z \in \mathit{SVars}(\mathcal{C}_{v}(E)) \setminus \{v\}$, $CaseIV(E,v,z)=\emptyset$, then there is no overlap between the spans extracted as $v$ and spans identified as contexts in $E$. 
\end{lemma}
\begin{proof}
Similar to the proof of Lemma~\ref{lm:case1}.
\end{proof}
\hide{
            \begin{definition}
            If $\forall v \in SVar(E)$ where $v \in V_u(\mathcal{P}_k)$, $E \in \mathit{Prov}(v)$, and $E$ is free of conflict $Case III$ if $CaseIV(E,v,q)=\emptyset$ $\forall q \in \mathit{SVars}(\mathcal{C}_{v}(E)) \setminus \{v\}$, then $E$ is free of \DefVar{Case IV conflicts}. 
            \end{definition}
            }

\begin{theorem}
Let $\mathcal{P}$ define an extractor with output spanner $\mathcal{P}_k$. If $\forall v \in SVar(E) \cap \mathit{SVars}(\mathcal{P}_k)$ where $E \in \mathit{Prov}(v)$ we have 
\begin{enumerate}
\item $CaseI(E,v)=\emptyset$,
\item $\forall z \in \mathit{SVars}(E) \setminus \{v\}$ $CaseII(E,v,z)=\emptyset$,
\item $\forall z \in \mathit{SVars}(E) \setminus \{v\}$ $CaseIII(E,v,z)=\emptyset$, and
\item $\forall z \in \mathit{SVars}(\mathcal{C}_{v}(E)) \setminus \{v\}$ $CaseIV(E,v,z)=\emptyset$
\end{enumerate}
then $\mathcal{P}$ is \DefVar{free of intra-formula conflicts}. 
\end{theorem}
\begin{proof}
    This summarizes Lemmas~\ref{lm:case1}-\ref{lm:case4}, and no other intra-formula conflicts are possible.
\end{proof}

\noindent \textbf{Inter-formula conflicts} \\An inter-formula conflict occurs if a span extracted by a formula $E \in \mathcal{P}$ is updated and shares a subspan that is matched by another formula $E' \in \mathcal{P}$. Given a pair of extraction formulas we construct the following spanner: 
$$InterOverlap({E},{E}',v)=  \bigcup\limits_{v'' \in \mathit{SVars}(S)}^{}  \pi_X (\rho_{v \rightarrow X} \spanner{E}) \join \spanner{\allen{X\: \cap \: Y}} \join \pi_Y ( \rho_{v'' \rightarrow Y} \spanner{S})
$$ where $E \neq E'$, $v \in SVar(E)$, $v' \in SVar(E') $ and $v'$ is an exposed variable, and $S=\mathcal{C}_{v'}(E')$.

\begin{lemma} \label{lm:interformula}
Given a pair of distinct extraction formulas $E$ and $E'$ if $\forall v \in \mathit{SVars}(E)$, $InterOverlap(E,$ $E',v)=\emptyset$ then the spans extracted by $E$ do not have any problematic overlap with spans consumed by $E'$. 
\end{lemma}
\begin{proof}
Similar to the proof of Lemma~\ref{lm:case1}.
\end{proof}

\begin{theorem}
 Given any extraction formula $E \in \mathcal{P}$, if $\forall v \in \mathit{SVars}(E) \cap \mathit{SVars}(\mathcal{P}_k)$ where $E \in \mathit{Prov}(v)$ the spans extracted as $v$ do not overlap with spans consumed by any other extraction formulas in $\mathcal{P}$ then $\mathcal{P}$ is \DefVar{free of inter-formula conflicts}.
\end{theorem}
\begin{proof}
    This results from Lemma~\ref{lm:interformula}, which is the only possible form of inter-formula conflict.
\end{proof}

\begin{definition}\label{def:confree}
An AQL program $\mathcal{P}$ is \DefVar{conflict-free} if its extraction formulas is free of both intra- and inter-formula conflicts. 
\end{definition}

 \begin{algorithm} 
\caption{Respecting Characters Test.}
\label{alg:respectChars} 
\LinesNumbered
\SetKwInput{Pre}{Precondition}
 \KwIn{Conflict-free core AQL program $\mathcal{P}$}
\KwOut{Boolean}
\Pre{All formulas has an exposed variable }
 $\mathcal{C} \gets \emptyset$
 
 \tcc{for all extraction formulas in $\mathcal{P}$}
\ForAll{$E \in \mathcal{P}$} {\label{line:everyE}

 $v \gets SVars(E).getExposed()$ \tcc{pick an exposed arbitrary variable}

 $\mathcal{U} \gets unigrams(\mathcal{C}_v(E))$   \tcc{contextualize $E$ by $v$; get all uncovered regions}
 
\ForAll{$u \in \mathcal{U}$}{

 $\mathcal{C} \gets \mathcal{C}\cup (\Sigma \setminus getChars(u))$ \tcc{$getChars(u)$ returns character set of $u$}\label{line:restrictedchars}
}
}

\ForAll{$V_i \in \mathit{SVars}(\mathcal{P}_k)$}{
 $rgxLs \gets emptyList()$ 
 
\ForAll{$E \in \mathit{Prov}(V_i)$}{

 $rgxLs.add(getEnclosedRegExs(E, V_i))$ \label{line: getexpression}\tcc{get regular expressions enclosed by $V_i$ in $E$}
}
\ForAll{$r\in rgxLs$} {
 $ \mathcal{G} \gets getChars(r)$ \label{line:getgoodcharacter}
 
\If{$\mathcal{G} \cap \mathcal{C} \neq \emptyset$}{
  \Return{$False$}
}
}
}
\Return{$True$} \tcc{Update respects characters}
\end{algorithm}
\subsubsection{Character conflicts}
With a conflict-free program, we guarantee that update translation does not affect the marked regions of any extraction formulas.  Essentially, we allow the unmarked regions to be freely modified, with the exception that no restricted characters can be deleted or inserted in those regions. Given an extraction formula, for each instance of an expression of the form $\hat\Sigma^*$ that is not covered by the contextualization algorithm, we define $\Sigma \setminus \hat\Sigma$ as the \textit{set of restricted characters}. 
For an  AQL program, we denote the set of  all restricted characters as 
$\mathcal{C}$ which is  the union of all  restricted character sets 
corresponding to  all  $E_i \in \mathcal{P}$. 

\begin{definition}\label{def:respect}
Given  a conflict-free  program $\mathcal{P}$ and  the update model $U$,  \DefVar{$U$ respects $\mathcal{P}$'s characters} 
if  $\forall v \in \mathit{SVars}(\mathcal{P}_k)$ and  $\forall E \in \mathit{Prov}(v)$ we have:  $\sigma  \in v \Rightarrow \sigma \notin \mathcal{C}$.
\end{definition}
\noindent Thus, the update model respects $\mathcal{P}$'s characters  if it  neither deletes nor inserts a symbol in $\mathcal{C}$. 
 \begin{lemma}
Given a conflict-free program, Algorithm~\ref{alg:respectChars} determines if the update model respects programs' characters.
 \end{lemma}
\begin{proof}
For each extraction formula in the input program (Line~\ref{line:everyE}), the set of restricted characters from its unigrams is retrieved and added to $\mathcal{C}$ (line~\ref{line:restrictedchars}). In our proposed update model, the domain and range of each view update function are determined by the domains of the associated attributes in the AQL program (see Section~\ref{sec:updateModel}). Therefore, the regular expressions associated with the update variables are retrieved in Line~\ref{line: getexpression}, and their corresponding character sets are extracted in Line~\ref{line:getgoodcharacter} and added to $\mathcal{G}$. If there is a common character between $\mathcal{G}$ and $\mathcal{C}$, the update may insert or remove a restricted character in a region marked as a unigram, which could adversely affect its matching. 

\end{proof}

\subsubsection{Join}

\hide{
            \begin{definition}
            A join operator is \DefVar{one-to-one} if 1) the set of join keys is non-empty, and 2) in each operand, all attributes are functionally dependent on the join keys.
            \end{definition}
}
\begin{definition}\label{def:nonexpjoin}
An arbitrary core AQL program is called \DefVar{non-expanding} if all join operators in the tree of extraction spanners are one-to-one.  
\end{definition}
\noindent If a join operator is not one-to-one, several tuples from one operand match the join attributes of some tuple from the other operand. Thus, if all attributes in both operand relations are functionally dependent on the join attributes, the join operator must be one-to-one.  
\begin{lemma}
Given a (non-empty) set of functional dependencies, if all join operators in an AQL program are necessarily one-to-one, regardless of input documents, (i.e., if all attributes in each pair of relations to be joined are functionally dependent on the join attributes) then for each operand relation there exists a partial identity function mapping its cells to cells in the resulting relation. 
\end{lemma}
\begin{proof}
Proof by contradiction: Consider $ Z = X \join Y $ where the join is one-to-one. Assume that there exists a cell $ x $ in $ r_x \in X $ that maps to at least two distinct cells in $ Z $, namely $ z_1 $ and $ z_2 $. This requires that $r_x$ is paired with more than one row in $ Y $. For this to happen, there are two possibilities: 1) At least two rows in $ Y $ have the same values for the join keys as in $r_x$, which contradicts the functional dependency condition for the attributes of $Y$; 2) The join is a Cartesian product, which cannot happen if the functional dependencies hold unless both operands are necessarily singleton sets, which contradicts that there is more than one row in $Y$. Thus, both possibilities lead to contradictions, and the assumption must be false. 
\end{proof}
 
\subsubsection{String Selection} Predicates involving at least one variable in a binary string selection within the tree of extraction spanners can be invalidated or validated by new values, causing some rows to appear or disappear, which is undesirable. As an additional sufficient condition, we require that all binary string selection operators must not have predicates over any update variables. 
\begin{definition}\label{def:restrictedselection}
   An arbitrary core AQL program $\mathcal{P}$ has \DefVar{restricted string selection} if for all binary string selection operators, ${\zeta^{=}_{x,y}}$ in $\mathcal{P}$'s tree of extraction spanners  $\{x,y\} \cap \mathit{SVars}(\mathcal{P}_k)=\emptyset$.
\end{definition}

Next, we prove that the introduced five conditions   (Definitions~\ref{def:domaincons},~\ref{def:confree},~\ref{def:respect},~\ref{def:nonexpjoin}, and~\ref{def:restrictedselection}) provide sufficient criteria for guaranteeing the stability of a core AQL program with respect to the extracted view update model (see Section~\ref{sec:updateModel}). 
    
\begin{theorem} \label{thrm:closure}
 An arbitrary core AQL program $\mathcal{P}$ is stable with respect to an update model $U$ that respects $\mathcal{P}$'s characters if $\mathcal{P}$ 
 \begin{enumerate}
     \item is domain-consistent, 
     \item is conflict-free,
     \item  is non-expanding, and
     \item  has restricted string selection. 
 \end{enumerate}
\end{theorem}
\begin{proof}

The proof is by contradiction. Assume that the value of a cell, $c$,  within column $A_i$ and row $r$ is updated from $v$ to $v'$ and the source document is updated accordingly. Despite meeting all four conditions, assume that $\mathcal{P}$ is not stable. If we execute the extraction program over $DSyn(D,r,i,v')$ two cases might occur:
\begin{enumerate}
\item \label{item:localconflict} $r$ has changed in an unexpected manner. Unexpected effects might be one or more of the following possibilities: 

\begin{enumerate}
\item At least one cell other than $c$, say $c'$ of column $A_j$, has changed in $r$. This implies that the update to $c$ has modified  the contextual, unigram, or  extracted region associated with $c'$  in at least one of $A_j$'s corresponding extraction formulas, i.e.,  $Prov(A_j)$;  but $\mathcal{P}$ is conflict-free and the update respects $\mathcal{P}$'s characters, so this case is not possible.
\item Cell $c$ has a value other than $v'$. Since $v'$ has  the same domain as $v$,  $\mathcal{P}$ has a consistent domain, and the update respects $\mathcal{P}$'s characters,  the match  associated with $r$ before the update must match after the update.  Therefore, this case is not possible. 
\item  Row $r$ has disappeared. For this to occur, either: \\
1) At least one of the corresponding extraction formulas in $Prov(A_i)$ does not match the region when  $v'$ appears in place  of $v$; but this case is impossible since the domain of $\mathcal{P}$ is consistent and $v$ and $v'$ have identical domains,  $\mathcal{P}$ is free of conflicts, and the update respects $\mathcal{P}$'s characters, so $v'$ could not have changed the context of any extraction formulas associated with $r$, or
2) extraction formulas responsible for matches of $r$ still match with new value $v'$, but at least one of them has disappeared because  of applying a binary selection operator somewhere in the tree: since $A_i$ is not involved in any predicate pertaining to a binary selection, this case is not possible.   
\end{enumerate}

\item  A row other than $r$, say $r'$,  has changed after the update. There are three possibilities:
\begin{enumerate}

\item \label{item:insertion} $r'$ is a new row that  has appeared as a result of the update.   The value of one or more extracted cells in $r'$ or the underlying value of their associated context or unigrams  of correpsonding extraction formulas   essential in forming the match of $r'$   did not form a match in $D$  but does form  one in $DSyn(D,r,i,v')$.  Because the extractor is free of conflicts and the update respects characters of $\mathcal{P}$, the update could not have modified the values of those  underlying values; thus this scenario is not possible. Furtheremore, the old value $v$ was associated to one cell in the output view but the new value $v'$ corresponds to two cells due to joins. This case is also not possible since exsiting joins are non-expanding. 
\item $r'$ is a row that  has disappeared due to the update:
The value of one or more extracted cells in $r'$ or the underlying value of their associated context or unigrams in correpsonding extraction formulas    essential in forming the match of $r'$ matched in $D$  but does not match in $DSyn(D,r,i,v')$.  Because the extractor is free of conflicts, $v'$ could not have modified the values of those  underlying values, and the update  respect characters; thus this scenario is not possible.
\item $r'$ is a row that has at least one extracted value changed due to the update. Because the extractor is free of conflicts $v'$ could not have modified the extracted values. 
\end{enumerate}
\end{enumerate}
 \end{proof}

Therefore, as long as we restrict ourselves to stable extractors and domain-preserving updates over extracted relations, we can apply any desirable cleaning process directly to extracted views and automatically update the documents through a straightforward translation. The clean documents can then be used in various data processing pipelines.  
\section{Experiments}
\subsection{Dataset}
Previous researchers have used data cleaning techniques on unstructured medical data in their pipelines~\cite{DBLP:conf/cbms/DeshpandeRTFRA20, woo2019application}. Unfortunately, the experimental datasets used in those papers are not freely available; so we instead implement similar techniques over the I2B2 dataset~\cite{noauthor_I2B2:_nodate}, which was initially created for evaluating two NLP challenge tasks~\cite{DBLP:journals/jamia/UzunerLS07}. We split the discharge summary file, \verb|unannotated_records_deid_smoking.xml|, into $889$ XML files, each of which represents an individual patient's discharge summary. We name  each file as \verb|N.xml|, where \verb|N| is the unique ID taken from \verb|<RECORD ID=N>| and uniquely identifies a patient's record. Each file has various sections, including \verb|Family history|, \verb|ALLERGIES|, \verb|Admission Date|, \verb|FOLLOW-UP INSTRUCTION|, etc. Although our proposed approach is independent of any structure in the data, the marking of these predefined sections simplifies and improves the accuracy of our extractors. An example of a constructed file is shown in Figure~\ref{fig:record}. Each file may vary slightly in structure; for example, not all records include a section for FOLLOW-UP INSTRUCTIONS. For simplicity, we focused our experiments on identifying and capturing only the most common patterns when designing our example extractors. Similar approaches could be used to clean the remaining data.

\begin{figure}
    \setlength{\fboxsep}{3pt} \centering
    \fbox{\begin{minipage}{0.50\textwidth}
             \sffamily \small
\textbf{<RECORD ID="547">}\\
\textbf{<TEXT>} \\
506243692\\
\textbf{FIH}\\
6305145\\
$\dots$\\

\textbf{ADMISSION DATE :}\\
09/11/2002\\
\textbf{DISCHARGE DATE :}\\
09/15/2002\\
$\dots$\\
\textbf{FOLLOW-UP INSTRUCTIONS :}\\
The patient will be seen in follow-up  by Dr. Warm one week post-discharge and by Gie Lung in the Fairm of $\dots$\\
\textbf{DD :}\\
09/15/2002\\
$\dots$\\
\textbf{</TEXT>}\\
\textbf{</RECORD>}
             \rmfamily
    \end{minipage}}\caption{Excerpt from the constructed file for record ID 547.}\label{fig:record}
    \end{figure} 

\subsection{Updatable Extracted Views as Data Cleaning Mechanisms}

Deshpande et al. have shown the significance of document cleaning in improving the performance of a retrieval system~\cite{DBLP:conf/cbms/DeshpandeRTFRA20}. Three cleaning processes were applied, including \textit{correcting errors and inconsistent values}, \textit{inserting missing data,}  and \textit{unifying abbreviations}, all of which are supported by our framework. 

In this section we present seven core AQL programs to demonstrate how our proposed method can be applied to clean the I2B2 dataset. We first present the primitive constructs used in our AQL programs, introduce the semantics of each cleaning extractor,  discuss the verification results, and provide statistics on each program.

Of course, the example extractors are designed to address the specific data quality issues found in this dataset, such as the difference in formatting between birth dates and admission dates. For other datasets, alternative extractors would need to be developed to address their particular quality issues. 

\subsubsection{Primitive Constructs} For readability, the  extractors are built on top of some simple grammar building blocks: 

\begin{tabular}{l l l}
\hspace{-0.5cm} \textbf{character sets:} \\ \small
$\gamma_{u}= [\regtext{A},\regtext{Z}]$ & $\gamma_{l}= [\regtext{a},\regtext{z}]$ & $\gamma_{d}=[\regtext{0},\regtext{9}]$ \\
$\gamma_b = \regtext{\smallunderscore} \vee \regtext{,}$ &  \multicolumn{2}{l}{$\gamma_{p} = (\regtext{:} \vee \regtext{,} \vee \regtext{;}  \vee \regtext{!} \vee \regtext{.} \vee \regtext{?})$} \\
$\gamma_{delim} = ( \regtext{/} \vee \regtext{-})$ & $\gamma_{s} = (\Sigma - \regtext{?}-\regtext{.})$ \\
 \end{tabular}

\begin{tabular}{r@{ } l r@{ } l}
\multicolumn{3}{l}{\hspace{-0.5cm} \textbf{basic patterns:}} \\ \small 
$\gamma_{\mathit{nd}}=$ & $\overbrace{\gamma_{d} \cdots\gamma_{d}}^n$ \\
$\gamma_{\mathit{date1}}=$&$\gamma_{\mathit{4d}} \gamma_{\mathit{delim}} \gamma_{\mathit{2d}} \gamma_{\mathit{delim}} \gamma_{\mathit{2d}} $ & $\gamma_{\mathit{date2}}=$&$\gamma_{\mathit{2d}}\gamma_{\mathit{delim}} \gamma_{\mathit{2d}}\gamma_{\mathit{delim}} \gamma_{\mathit{4d}}  $ \\
$\gamma_{\mathit{date3}}=$&$\gamma_{2d}\gamma_{\mathit{delim}} \gamma_{2d} \gamma_{\mathit{delim}} \gamma_{5d}$ & $\gamma_{\mathit{Date}}=$ & $\gamma_{\mathit{8d}} \vee \gamma_{\mathit{date1}} \vee \gamma_{\mathit{date2}} \vee \gamma_{\mathit{date3}} $ \\
$\gamma_{\mathit{Abr}}=$&$\gamma_{u}\gamma_{u}\gamma_{u}^*$ &$\gamma_{Comp}$=&$\gamma_{u}\gamma_{l}\gamma_{l}*$($\regtext{\smallunderscore}\gamma_{u}\gamma_{l}\gamma_{l}^*)^*$\\
$\gamma_{\mathit{Word}}=$&$(\gamma_{u} \vee \gamma_{d})(\gamma_{u} \vee \gamma_{d})^*$\\
  \\
\end{tabular} \normalsize

\noindent Finally, for some of the extractors, we use a few small word lists: 
\small
\begin{align*}
    \gamma_{verb} &= \regtext{recommend} \vee \regtext{recommended}  \vee \regtext{have} \vee \regtext{had} \vee \regtext{receive} \vee \regtext{received} \vee \regtext{perform} \vee \regtext{performed} \\
    \gamma_{sec1}&=\regtext{Discharge\smallunderscore Date} \vee \regtext{DISCHARGE\smallunderscore DATE} \\
    \gamma_{sec2}&=\regtext{MEDICATIONS\smallunderscore ON\smallunderscore DISCHARGE} \vee \regtext{Discharge\smallunderscore Medications} \vee \regtext{DISCHARGE\smallunderscore MEDICATIONS}
\end{align*}
\normalsize
\begin{table*}[btp] \caption{AQL Programs for cleaning I2B2 Dataset. Update variables are highlighted in \textbf{bold}.}   \label{tbl:aql}
\centering 
\begin{adjustbox}{width=1\textwidth} 
\renewcommand{\arraystretch}{1.5} 

\begin{tabular}{r@{ } l }  
         \hline 
\multicolumn{2}{l}{\textbf{Date Formatting Consistency and Correctness}} \\

$E_1=$ & $\llbracket \Sigma^* (\regtext{\smallunderscore} \vee \gamma_p) (\regtext{admitted} \vee  \regtext{Admitted\smallunderscore on}) \regtext{\smallunderscore} \textbf{D}\{\gamma_{\mathit{Date}}\}(\regtext{\smallunderscore} \vee \gamma_p)\Sigma^* \rrbracket$ \\
$E_2=$& $ \llbracket \Sigma^* (\regtext{ADMISSION\smallunderscore DATE}\vee\regtext{Admission\smallunderscore Date }) \regtext{\smallunderscore:\textbackslash n}\textbf{D}\{\gamma_{\mathit{Date}}\}\regtext{\textbackslash n}\Sigma^*\rrbracket$\\ 
$E_3=$& $ \llbracket \Sigma^* (\regtext{DISCHARGE\smallunderscore DATE}\vee\regtext{Discharge\smallunderscore Date }) \regtext{\smallunderscore:\textbackslash n}\textbf{D}\{\gamma_{\mathit{Date}}\}\regtext{\textbackslash n}\Sigma^*\rrbracket$\\ 

$E_{\mathit{date}}=$& $  E_1 \cup E_2 \cup E_3 $\\
\hline
\multicolumn{2}{l}{\textbf{Age Correctness}} \\

$E_{4}=$&$ \llbracket \Sigma^* \regtext{\smallunderscore }\textbf{A}\{\gamma_{d} \vee \gamma_{2d} \vee \gamma_{3d} \}(\regtext{\smallunderscore} \vee \regtext{-}) \regtext{year} (\regtext{\smallunderscore} \vee \regtext{-})\regtext{old\smallunderscore}\Sigma^*\rrbracket$ \\
$E_{\mathit{age}}=$& $  E_4  $\\
\hline

\multicolumn{2}{l}{\textbf{Value Consistency}} \\

$E_{5}=$& $\llbracket\Sigma^{*} 
   \regtext{<RECORD\smallunderscore ID=``}R\{\gamma_d\gamma_{d}^*\}\regtext{''>\textbackslash n<TEXT>\textbackslash n}(\Sigma-\regtext{<})^* \regtext{\textbackslash n}T1\{\gamma_{sec1}\}\regtext{\smallunderscore:\textbackslash n}  \textbf{D1}\{\gamma_{Date}\} \regtext{\textbackslash n}\Sigma^*\rrbracket$\\
    $E_{6}=$&$\llbracket\Sigma^{*} 
     \regtext{<RECORD\smallunderscore ID=``}R\{\gamma_d\gamma_{d}^*\}\regtext{''>\textbackslash n<TEXT>\textbackslash n}(\Sigma-\regtext{<})^* \regtext{\textbackslash n}T2\{\regtext{DD} \vee \regtext{D}\}\regtext{\smallunderscore:\textbackslash n}  \textbf{D2}\{\gamma_{Date}\} \regtext{\textbackslash n}\Sigma^*\rrbracket$\\
$E_{\mathit{valCon}}=$& $E_{5} \join E_{6} $\\
\hline
\multicolumn{2}{l}{\textbf{Order Dependency}} \\

$E_7=$&$\llbracket \Sigma^* \regtext{<RECORD\smallunderscore ID=``}R\{\gamma_d\gamma_{d}^*\}\regtext{''>\textbackslash n<TEXT>\textbackslash n}(\Sigma- \regtext{<})^*(\regtext{\smallunderscore} \vee \gamma_p) A\{ (\regtext{admitted} \vee  \regtext{Admitted\smallunderscore on}) \}\regtext{\smallunderscore} \textbf{D1}\{\gamma_{\mathit{Date}}\}(\regtext{\smallunderscore} \vee \gamma_p)\Sigma^* \rrbracket$\\
   $E_8=$&$\llbracket \Sigma^* \regtext{<RECORD\smallunderscore ID=``}R\{\gamma_d\gamma_{d}^*\}\regtext{''>\textbackslash n<TEXT>\textbackslash n}(\Sigma- \regtext{<})^* A\{(\regtext{ADMISSION\smallunderscore DATE}\vee\regtext{Admission\smallunderscore Date }) \}\regtext{\smallunderscore:\textbackslash n}\textbf{D1}\{\gamma_{\mathit{Date}}\}\regtext{\textbackslash n}\Sigma^*\rrbracket$\\
$E_9=$&$ \llbracket \Sigma^* \regtext{<RECORD\smallunderscore ID=``}R\{\gamma_d\gamma_{d}^*\}\regtext{''>\textbackslash n<TEXT>\textbackslash n}(\Sigma- \regtext{<})^*S\{(\regtext{DISCHARGE\smallunderscore DATE}\vee\regtext{Discharge\smallunderscore Date }) \}\regtext{\smallunderscore:\textbackslash n}\textbf{D2}\{\gamma_{\mathit{Date}}\}\regtext{\textbackslash n}\Sigma^*\rrbracket $\\
$E_{order}=$&$E_9 \join (E_7 \cup E_8)$\\
\hline
\multicolumn{2}{l}{\textbf{Medication List Formatting}} \\
$E_{10}=$&$\llbracket\Sigma^{*} \gamma_{sec2}\regtext{\smallunderscore:\textbackslash n}(\gamma_d \vee \gamma_{u} \vee \gamma_{l} )  (\Sigma-\regtext{:})^* S\{(\regtext{\smallunderscore} (\regtext{,}\vee \regtext{;})\regtext{\smallunderscore} )\vee \regtext{\textbackslash n} \}(\gamma_{u} \vee \gamma_{l} \vee \gamma_{d})  \Sigma^*\rrbracket$\\
$E_{list}=$&$E_{10}$\\
\hline
\multicolumn{2}{l}{\textbf{Missing Units}} \\

$E_{11}=$&$\llbracket\Sigma^{*} \gamma_{sec2}\regtext{\smallunderscore:}(\Sigma-\regtext{:})^*\regtext{\textbackslash n} M\{ \gamma_{Comp}\}\regtext{\smallunderscore} G\{\gamma_d\gamma_{d}^*\}  U\{\regtext{\smallunderscore} (\gamma_l \vee \regtext{\smallunderscore} \vee \gamma_d \vee \regtext{.})^* \} \regtext{\textbackslash n} \Sigma^*\rrbracket$\\
$E_{unit}=$&$E_{11}$\\
\hline
\multicolumn{2}{l}{\textbf{Unifying Abbreviations and Full Forms} } \\
$E_{12}=$& $\llbracket \Sigma^* \regtext{\smallunderscore} \gamma_{verb} \regtext{\smallunderscore}(\gamma_{\mathit{Word}}\regtext{\smallunderscore})^*\textbf{U}\{\gamma_{\mathit{Abr}} \vee \gamma_{\mathit{{Comp}}}\} \gamma_b(\gamma_{p} \vee \gamma_l \vee \gamma_d) \Sigma^*\rrbracket$\\ 
$E_{\mathit{unify}}=$& $E_{12}$\\
\hline

    \end{tabular}
\end{adjustbox}
    \label{fig:uvx_extractors}
\end{table*}

\normalsize
\subsubsection{Data cleaning Semantics} We address potential data quality issues commonly found in free-text electronic medical records: 
 
\begin{description}
    \item [Date Formatting Consistency and Correctness] We observe inconsistencies in the formatting of dates in various places. Some dates are expressed as separated and grouped six-digit patterns, while others are simple eight-digit strings. Different separators, such as hyphens ($\regtextinl{-}$) and slashes ($\regtextinl{\slash}$), have been used. While we have not observed any errors in the date values, such as having five-digit years, our program handles some potential inaccuracies in values as well.  
    \item [Age Correctness] We have not observed any out-of-range patient age values, such as $168$, in the I2B2 dataset. Since the dataset is de-identified, it is possible that such errors were eliminated during the anonymization process. Nevertheles, we also illustrate how to clean any such erroneous values, should they appear.

    \item [Value Consistency]  A patient's discharge date might be recorded in two positions within each record, namely in the \verb|DD|  or \verb|Discharge Date| sections. These values should be consistent, and any discrepancies must be corrected.

    \item [Order Dependency] There is an inherent order dependency~\cite{DBLP:journals/pvldb/SzlichtaGG12} in the dataset: within each record, the discharge date cannot be earlier than the admission date.
    $$\textit{admission date} \mapsto \textit{discharge date}$$
    We illustrate how to resolve any violations of this dependency. 

    \item [Medication List Formatting] The list of medications prescribed at discharge is recorded inconsistently, sometimes as a structured list and other times as a sequence of words. We illustrate how to standardize the formatting across all records by presenting all discharge medications as proper list items. 

    \item [Missing Units]
    We also illustrate how to identify and correct instances where the dosage unit, e.g., mg (milligram), of a prescribed medication is not recorded at the time of discharge.\footnote{It is reported that nearly 50\% of all medication errors occur during the prescribing or ordering process, which can elevate the risk of patient mortality~\cite{tariq2018medication}.}
    
    \item [Abbreviation] Finally, the dataset contains many instances  of medical terms in both abbreviated and full forms; for example, we find  \regtextinl{...had\smallunderscore been\smallunderscore on\smallunderscore Macrodantin\smallunderscore suppressive\smallunderscore therapy\smallunderscore preoper} \regtextinl{atively...} 
   as well as
    \regtextinl{...recommended\smallunderscore that\smallunderscore an\smallunderscore IVP} \regtextinl{\smallunderscore be\smallunderscore obtained...}.  For downstream tasks, it may be beneficial to unify these instances (either to all full forms or to abbreviations)~\cite{DBLP:conf/cbms/DeshpandeRTFRA20}. 
\end{description}

\subsubsection{Cleaning Extractors} The core AQL programs designed to address the identified data quality problems in the I2B2 dataset are presented in Table~\ref{tbl:aql}. In each program, we denote extraction formulas as $E_k$, where $k$ is an integer, and output spanners as $E_s$, where $s$ is a string. For example the \textbf{Date Formatting Consistency and Correctness} program comprises three extraction formulas $E_1$, $E_2$, and $E_3$. The  output spanner, $E_{date}$, is expressed as the union of the three extraction formulas   where  $\textbf{D}$ is the program's only update variable. 

\subsubsection{Verification and Cleaning}

If the proposed verifier determines that a program is stable, any domain preserving update made to the corresponding extracted views can be transfered back to their original locations in the source documents. Therefore, any data cleaning method applied to the extracted table can be translated back to the originating documents, as long as the method is domain preserving, effectively cleaning the text data once and for all. 
 \begin{description}
     \item[Date Formatting Consistency and Correctness]   The verifier finds this program stable. A total of 950 instances were matched by $E_{\mathit{date}}$ across 487 matched documents. These documents exhibit two distinct formatting categories. Table~\ref{tbl:dateformatting} shows two of the extracted records along with their document IDs and corresponding spans. If our goal is to clean the documents by standardizing the date format to \verb|yyyy-mm-dd|, after transferring the 950 updated values back to their corresponding locations in the source documents, running $E_{\mathit{date}}$ over the updated documents results in extracting records that include those shown in Table~\ref{tbl:dateformatting_reextraction}, as expected.

    \item[Age Correctness] This is also a stable program. $E_{age}$ identified 513 age instances across 462 matched documents. As noted earlier, all identified instances fall within the valid age range, so no cleaning is required.
    
    \item [Value Consistency:]   This program passes four verification tests, namely the program is domain-consistent and conflict-free, its corresponding update translation respects characters, and it does not have any problematic string selection operator.  The presence of a join in the tree of extraction spanners  requires us to examine whether the program is non-expanding as well. The set of join keys $\{\textbf{R}\}$ is non-empty, and it is known that  all other attributes, such as $\textbf{D1}$, $\textbf{T1}$,  $\textbf{D2}$, and $\textbf{T2}$,  are functionally dependent on $\textbf{R}$.  Therefore, the program is found to be stable. A total of 107 records are identified by $E_{\mathit{valCon}}$ across 107 documents. Whichever values a cleaning program run on the table chooses to update, they can be transferred back to the documents to remove those inconsistencies.
    
    \item [Order Dependency] Admission dates are explicitly recorded in two places: in the \verb|Admission| \verb|Date| section (handled by $E8$) and within the explanations of other sections, such as \verb|Hospital Course| \verb|and Treatment| (identified using $E7$). In the case of E8, it is known that the dependency of $\textbf{A}$ (and $\textbf{D1}$ ) on $\textbf{R}$ is not functional. Therefore, despite passing the first four tests, $E_{order}$ cannot be shown to be non-expanding and, consequently, it is not stable. However, we can accomplish the same cleaning goal by modifying the variable names in $E_9$ to make it union compatible with $E_7$ and $E_8$ and redefine the output spanner as $E'$ where $E'= E_7 \cup E_8  \cup E_9$. Since $E'$ is determined to be stable.\footnote{An interesting research question that arises from this case is whether stable programs containing joins can always be expressed using other algebraic operators.} A total of 487 documents are matched using $E'$. With the exception of 24 cases where our extractors do not match the admission date or do not match the discharge date, all other instances respect the dependency rule.
    
    \item [Medication List Formatting] This program is detemined to be stable. A total number of 489 document match $E_{list}$ from which 2602 relational records can be extracted. There are 1,252 instances where we needed to replace \regtextinl{,} and \regtextinl{;} with \regtextinl{\textbackslash n} to create a proper list of medications. We emphasize once again that more accurate extractors could be designed to capture a greater range of data quality problems.

    \item [Missing Units] By extracting each numeric quantity and the remainder of the medication line (and assuming that \textbf{Medication List Formatting} has been corrected), we can identify whether the word following the quantity represents some dosage units, such as \regtext{mg} or \regtext{tsp}. If not, we can insert the units by updating the corresponding \textit{U} cell. 
    A total of 113 documents matched $E_{unit}$, resulting in 304 extracted records, of which 18 were missing the specified units. 
    (It is perhaps more intuitive to extract the units if present---i.e., the next word only if it represents a valid unit abbreviation---and otherwise extract the empty string $\epsilon$. However, the verifier would reject such an extractor for being unstable, because $CaseIV$ would uncover a problematic overlap between regions marked as $\textbf{U}$ and those marked as context.) 

    \item[Abbreviation] $E_{\mathit{unify}}$  identifies 922 instances of medical terms in both full and abbreviated forms across 396 documents. There are instances, such as \regtextinl{Clinchring\smallunderscore Health\smallunderscore Center} or \regtextinl{February}, that are identified due to the lower precision of our extractors but are not of interest. These instances would remain unchanged by a cleaning program. Some of the extracted and potentially cleansed records are depicted in Tables~\ref{tbl:unifyoriginal} and \ref{tbl:unifycleaned}.  Note that \regtextinl{Intensive\smallunderscore Care\smallunderscore Unit} extracted form 819.xml need not be updated if full forms are preferred, however  updating \regtextinl{CT} changes the length of text comming before it, consequently the span for \regtextinl{Intensive\smallunderscore Care\smallunderscore Unit} in the re-extracted relation is shifted.

 \end{description}
 
\begin{table*}[tb] 
\caption{ These tables show some of the extracted records from uncleaned documents  on the left side, and their corresponding re-extracted records from cleaned documents  on the right side. The name of the cleaned version of a document is in \textit{Italic}. 
} \label{tab:experiment}
\begin{subtable}{.4\textwidth}
\caption{Extracted Records for $E_{\mathit{date}}$ } \label{tbl:dateformatting}
    \centering
    \begin{adjustbox}{width=0.8\textwidth} 
    \begin{tabular}{@{}c|c|ccc} 
Doc	&	Span	&	String Value		\\
\hline										
45.xml	&	$[261,269\rangle$	&	\regtext{20050305}	\\
	&	$[235,243\rangle$	&	\regtext{20050301}	\\
\hline
701.xml	&	$[216,226\rangle$	&	\regtext{05/16/2004}	\\
	&	$[244,254\rangle$	&	\regtext{05/26/2004}	\\
	&	$[2290,2300\rangle$	&	\regtext{05/16/2004}	\\
 \end{tabular}
    \end{adjustbox}
\end{subtable}
~
\begin{subtable}{.4\textwidth}
\caption{Re-extracted Records for $E_{\mathit{date}}$ } \label{tbl:dateformatting_reextraction}
    \centering
    \begin{adjustbox}{width=0.76\textwidth} 
  \begin{tabular}{@{}c|c|ccc} 
Doc	&	Span	&	String Value		\\
\hline										
\textit{45.xml}	&	$[263,273\rangle$	&	\regtext{2005-03-05}	\\
	&	$[235,245\rangle$	&	\regtext{2005-03-01 }	\\
\hline
\textit{701.xml}	&	$[216,226\rangle$	&	\regtext{2004-05-16} 	\\
	&	$[244,254\rangle$	&	\regtext{2004-05-26}	\\
	&	$[2290,2300\rangle$	&	\regtext{2004-05-16} 	\\
 \end{tabular}
    \end{adjustbox}
\end{subtable}
    \bigskip

\begin{subtable}{.4\textwidth}
\caption{Extracted Records for $E_{\mathit{unify}}$ } \label{tbl:unifyoriginal}
    \centering
    \begin{adjustbox}{width=0.9\textwidth} 
    \begin{tabular}{@{}c|c|ccc} 
Doc	&	Span	&	String Value		\\
\hline										
43.xml	&	$[1506,1508\rangle$ &	\regtext{CT}	\\
	&	$[702,705 \rangle$ &	\regtext{MRI}	\\
\hline										
819.xml	&	$[1035,1037\rangle$ &	\regtext{CT}	\\
	&	$[1305,1324\rangle$ &	\regtext{Intensive\smallunderscore Care\smallunderscore Unit}	\\
\hline
575.xml	&	$[2073,2084\rangle$ &	\regtext{Pseudomonas}	\\
	&	$[2474,2477\rangle$ &	\regtext{DVT}	\\
	&	$[2518,2521\rangle$ &	\regtext{TPA}
    \end{tabular}
    \end{adjustbox}
\end{subtable}
~
\begin{subtable}{.45\textwidth}
\caption{Re-extracted Records for $E_{\mathit{unify}}$} \label{tbl:unifycleaned}
    \centering
    \begin{adjustbox}{width=0.95\textwidth} 
    \begin{tabular}{@{}c|c|ccc} 
Doc	&	Span	&	String Value		\\
\hline										
\textit{43.xml}	&	$[1529,1548\rangle$ &	\regtext{Computed\smallunderscore Tomography}\\
	&	$[702,728 \rangle$ &	\regtext{Magnetic\smallunderscore Resonance\smallunderscore Imaging}	\\
\hline										
\textit{819.xml}	&	$[1035,1054\rangle$ &	\regtext{Computed\smallunderscore Tomography}	\\
	&	$[1323,1342\rangle$ &	\regtext{Intensive\smallunderscore Care\smallunderscore Unit}	\\
\hline
\textit{575.xml}	&	$[2073,2084\rangle$ &	\regtext{Pseudomonas}\\
	&	$[2474,2494 \rangle$ &	\regtext{Deep\smallunderscore Vein \smallunderscore Thrombosis}	\\
&	$[2535,2563\rangle$ & \regtext{Tissue\smallunderscore Plasminogen\smallunderscore Activator}
    \end{tabular}
    \end{adjustbox}
\end{subtable}\end{table*}

\section{Related Work}

\subsection{Data cleaning}
Data cleaning plays a vital role in enhancing data quality by identifying and rectifying inconsistencies, errors, and redundancies. This process ensures that the data is more accurate and reliable, ultimately supporting better analysis and  decision-making~\cite{DBLP:reference/bdt/Chu19}.   Extensive research has explored various aspects of data cleaning, including addressing key issues by  deduplication, outlier detection, and logical errors~\cite{DBLP:books/acm/IlyasC19}, as well as enhancing the scalability of cleaning algorithms~\cite{DBLP:phd/basesearch/Saxena21}. Further efforts have focused on developing specialized algorithms tailored to the semantics of specific datasets such as time series~\cite{DBLP:journals/pvldb/DingWSLLG19}, medical~\cite{DBLP:conf/cbms/DeshpandeRTFRA20}, Chinese-language e-business~\cite{5917038}, and GPS  trajectory~\cite{DBLP:conf/dasfaa/LiCLB20} datasets; To achieve these objectives, various approaches have been employed, including  rule-based methods~\cite{DBLP:phd/us/Ebaid19, DBLP:conf/icde/ChuIP13}, probabilistic techniques~\cite{DBLP:conf/icde/QinHWZZM0O024,DBLP:journals/fcsc/LiLLC19}, or hybrid models combining both~\cite{DBLP:journals/pvldb/RekatsinasCIR17}. More sophisticated requirements, such as privacy considerations, have also been introduced into the data cleaning process~\cite{DBLP:journals/is/HuangMC20}. Additionally, some researchers have proposed incorporating user input into the cleaning process~\cite{DBLP:journals/jdiq/PereiraFLG24}. A large amount of data exists in semi-structured and unstructured formats. However, limited effort has been devoted to addressing the data quality issues inherent in these formats~\cite{DBLP:journals/pvldb/ChuI16}, which is the primary focus of our work.   Ilyas and Chu provide a comprehensive review of data cleaning~\cite{DBLP:books/acm/IlyasC19}.

\subsection{Updatability of Extracted Views} 

Expectations from extractors have risen as requirements have become more diversified, from the point that there were no  criteria to evaluate their performance~\cite{gaizauskas1998information} to the point that extraction algorithms need to work under various stresses such as noisy data, low response time, and diverse types of input and output~\cite{DBLP:journals/ftdb/Sarawagi08}. The body of research related to the problem of updatable views over unstructured data is relatively limited. However, we can draw insights from existing work on updatable views over semi-structured data, such as XML, which is slightly relevant to our study. Similar to the relational setting, creating views over XML databases can provide various advantages, including faster query processing and convenient access control over specific sections of a larger XML database \cite{DBLP:conf/vldb/AbiteboulMRVW98}. Kozankiewicz et al.~\cite{DBLP:conf/adbis/KozankiewiczLS03} propose to incorporate  information about forseeable updates over views into the  view definition. Therefore,  the ultimate affects of updates are specified solely by the query designer which, if not verified, might leave the XML database in an incorrect state .  
\subsection{Updatability of Relational Views}
Our research question involves translating updates on  the content of an extracted view to updates on the content of the associated  document. This is similar to the problem of updatability of relational views that is thoroughly studied in relational databases~\cite{DBLP:books/sp/kimrb85/FurtadoC85, DBLP:journals/is/FurtadoSS79,DBLP:journals/computer/Keller86,DBLP:journals/algorithmica/MedeirosT86, DBLP:journals/pvldb/MeliouGS11}. 

In the relational setting, the problem of a view updates is defined as finding a translation of a view update   to a database update  such that running the same view definition query on the updated relation produces the updated view regardless of the database instance. Interesting research  challenges are raised from this definition such as how to deal with the problem of multiple possible translations? Are views always updatable? if not how to identify views that cannot be updated? how to derive a specific translation mechanism for a given view definition, database schema, and update specification?

Two general approaches are proposed for updatability of relational views.  First, along with a view definition, all authorized updates and their corresponding translations should be provided. However, the provided translation mechanism also needs to be verified. The second approach is to exploit the information provided by the view definition, the update mechanism, and database constraints to derive conditions on the legitimacy of a translator. For example, the view dependency graph that is constructed using only the view definition and database schema is used to verify a translator for some classes of deletions, insertions, and replacements~\cite{DBLP:journals/tods/DayalB82}. Our solution to the extracted view updates is aligned with the latter approach.  We limit ourselves to a class of view updates that is realized by a domain-preserving function that maps each extracted value to a value from the same domain, similar to perturbing values of a table to protect privacy.  However,  we do not impose any constraints on the input documents. 
We pick the most natural translation which is to substitute old values in the source document with new values, and we expect to see them extracted by running the same extractor (Figure~\ref{fig:exvu}).\par
In summary, applying solutions proposed in the relational setting to the information extraction domain poses significant challenges. The relational setting benefits from various constraints, including schema-based constraints such as data types, referential integrity constraints, key constraints, and functional dependency constraints, among others. These constraints serve to structure and regulate the problem space. However, in the context of information extraction, such constraints are generally not present: there are usually no inherent limitations on the content of the source documents.

\subsection{Rule-based  versus LLM-based Extraction}\label{subsubsec:whyrulebase}

When using extractors as document cleaning tools, it is essential to trace updates back to the original documents and predict how changes in the documents will affect the extracted relations. We meet these two requirements through deterministic and computable extractors. While we can ensure these properties with rule-based extractors, it remains unclear how to enforce such guarantees in extractors based on pre-trained language models (PLMs), including large language models (LLMs):

\begin{description}
    
\item [Computable Extractor:] Being a strict extractor implies that the extracted items must occur in the input text, and computability requires that the extraction process is capable of generating the necessary provenance. Because, not all instances of a term or phrase in the source document may correspond to those that are extracted, we require a mechanism for identifying the corresponding positions, i.e., fine-grained data lineage. Extractors written in certain rule-based extraction languages, such as AQL (used in SystemT) and JAPE (used in GATE),  are inherently computable, i.e., the lineage of extracted items are available as a by-product of the extraction process. However, extractors expressed as PLMs do not come with these inherent capabilities. PLMs are  complex and   perceived as black-box solutions. As a result, a significant body of research is dedicated to inventing novel techniques to explain how PLMs operate~\cite{DBLP:journals/corr/abs-2309-01029}. The fine-grained lineage of extracted items can be considered a \textit{local explanation} mechanism,\footnote{For a comprehensive overview of explainability,  refer to~\cite{DBLP:journals/corr/abs-2309-01029}.} which aims to provide insight into how a model responds to a specific input instance. Several techniques belong to this category, among which explanations based on \textit{attribution} might offer a viable mechanism for pinpointing  corresponding positions in the input. An attribution-based explainer assigns a relevancy score to each input word, highlighting its contribution to generating the output (extracted items, in the extraction case). In our work,  computability is treated as a non-probabilistic property, whereas lineage based on attribution is inherently probabilistic~\cite{DBLP:journals/corr/abs-2305-06311}, so current PLMs do not meet our requirement for computability.

\item [Deterministic Extractor:] The generated outputs in PLMs are the outcomes of a combination of multiple stochastic and/or heuristic steps. Consequently, running a PLM multiple times for the same input can yield different responses~\cite{DBLP:journals/corr/abs-2307-15343,GPTRep}, which characterizes LLMs as non-deterministic extractors. Again they do not meet the requirements required for extracting updatable views.  

\end{description}

\subsection{Fine-grained Data Lineage} 
Data lineage, or provenance, has been defined and formalized for structured and semi-structured data~\cite{DBLP:conf/icdt/BunemanKT01,DBLP:journals/ftdb/CheneyCT09}. Given a value that is the outcome  of executing a well-defined query over some data sources, often relational tables, provenance determines three aspects related to the value: data points in the source that contribute to form the value, the way that data points collaborate to produce the value, and the exact location(s) in the data source from which the value originates. The last aspect is similar to the notion of lineage that we use in this work, i.e., we require the extractor to provide the positions in a document from which a value is extracted.

Provenance-based techniques have also been applied to information extraction problems. Roy et al.~\cite{roy2013provenance} propose a provenance-based technique to improve the quality of extraction by refining the dictionaries that are used in a rule-based extraction system. A set of entries from the dictionaries that have been involved in generating the output are analyzed to determine which should be removed to improve the extractor's performance most.  In other work, Liu et al.~\cite{liu2010automatic} use provenance techniques to determine the most effective rule refinements, i.e., those that result in removing undesirable tuples and keeping correct ones.
Chai et al.~\cite{DBLP:conf/sigmod/ChaiVDN09} examine the provenance of a multi-stage extraction program that can include relational operators on intermediate tables. Users' feedback is expressed as updates over tuples that appear at any stage of the extraction process, and these updates are translated into modifications of the corresponding extraction program.    

\subsection{Static Analysis of Programs Using Regular Languages}
 We use an extended form of finite-state automata to statically analyse an extraction program and the update mechanism. Similar static analyses of regular expressions or deterministic finite automata   have been used in diverse areas, including  access control, feature interactions, and vulnerability detection of programs. For example, Murata et al.~\cite{DBLP:journals/tissec/MurataTKH06} propose an automaton-based access control mechanism for XML database systems. Regular expressions are derived  from given queries, access-control policies, and  schemas. Based on the characteristics of the derived automata, element/attribute level access requests  by  queries are determined to be either  granted, denied, or statically indeterminate, independently of any actual input XML documents. An event-based framework is introduced by Kin et al.~\cite{DBLP:conf/chi/KinHDA12} for developing and maintaining new gestures that can be used in multi-touch environments. Using that framework, application developers express each gesture as a regular expression over some predefined touch events. Regular expressions associated with  gestures are then statically analyzed to identify possible conflicts between various gestures. Yu et al.~\cite{DBLP:journals/fmsd/YuABI14} present a method for detecting security vulnerabilities in programs that use string manipulation operators such as \emph{concatenation} and \emph{replacement}. In essence, their approach involves constructing DFAs to represent the program's data dependency graph. Subsequently, they perform static analysis on the provided attack pattern, expressed as a DFA, and the graph's DFAs to detect potential vulnerabilities.
 Dynamiclly generated SQL queries can enhance the flexibility of programs, written in  JAVA or other languages,  by allowing them to adapt to changing conditions and requirements without the need to write multiple static queries. However, the host compiler, i.e., JAVA compiler,  does not test the generated query strings for possible errors such as type errors.

 To this end, Wassermann et al.~\cite{DBLP:journals/tosem/WassermannGSD07} propose a static analyzer to verify the correctness of dynamically constructed SQL queries embedded in programs. Their approach involves creating a DFA representation of the generated query strings and   performing static analysis on the  DFA.

\section{Conclusions and Future Work}
We propose and develop a framework for systematically cleaning documents, utilizing specialized extractors that ensure consistency and predictability throughout the entire process.
We characterize extraction algorithms that are resilient to changes in their source documents intended to reflect predetermined changes to the extracted relations, i.e.,  \emph{stable} extractors. 
We further propose a straightforward algorithm that modifies the input document considering the stable extractor and a set of domain preserving functions. If the modified document is fed into the extractor, it will produce the expected updated relation. Moreover, we propose a verification process to ensure the stability of programs written in the AQL language. The verifier tests four sufficient conditions. Through experiments, we demonstrate that the sufficient conditions are highly likely to hold in practice.

 We have made some simplifying assumptions in our work, each of which can be modified or eliminated to expand the class of programs deemed to be stable. For example, we have assumed independence between extracted attributes, thus requiring that at most one extracted value can be affected by each change in the source document. What if instead we are given constraints among the attributes, such as $A_2$ and $A_5$ must be identical or must have (computably) dependent values? We have also assumed that each table attribute can be given a value from a single span in the input document. What if several words or phrases from multiple places can be combined to create an extracted value? Loosening our simplifying assumptions might result in being able to verify more useful cleaning extractors.

We have presented a property verification process applicable to programs expressed by a subset of AQL. 
We have designed the complete verification process and proved its correctness. However, AQL is a broad language with  many more  operators and features beyond what is covered in this work.  \par 
We have developed a set of sufficient properties for stable Core AQL programs, but we have not investigated which properties might be necessary. We might also wish to explore whether some programs that cannot be verified as stable can be transformed into ones that possess the required properties of stability.

A natural extension for this problem involves developing verification tools for other rule-based extraction languages, such as JAPE, in case those extractors fit more naturally with particular applications. 
Finally, we are also interested in developing cleaning tools that can be used with extractors that are not based on regular languages, including those based on machine learning technology.

\newpage

\bibliographystyle{plainurl}
\bibliography{IE-ref}

\end{document}